\documentclass[conference]{IEEEtran}

\IEEEoverridecommandlockouts
\usepackage{cite}
\usepackage{amsmath,amssymb,amsfonts}
\usepackage{algorithmic}
\usepackage{graphicx}
\usepackage{textcomp}
\usepackage{xcolor}
\usepackage{multicol}
\usepackage{multirow}
\usepackage{subfigure}
\usepackage{booktabs}
\usepackage{enumitem}
\usepackage{multirow}
\usepackage{threeparttable}
\usepackage[ruled,vlined]{algorithm2e}
\SetKw{Continue}{continue}
\usepackage{flushend}
\usepackage{hyperref}
\usepackage[misc]{ifsym}
\usepackage{amsthm}
\usepackage{balance}

\newtheorem{example}{Example}
\newtheorem{lemma}{Lemma}
\newtheorem{definition}{Definition}

\SetCommentSty{mycommfont}
\def\BibTeX{{\rm B\kern-.05em{\sc i\kern-.025em b}\kern-.08em
    T\kern-.1667em\lower.7ex\hbox{E}\kern-.125emX}}
\begin{document}

\title{MultiEM: Efficient and Effective Unsupervised Multi-Table Entity Matching}

\author{\IEEEauthorblockN{Xiaocan Zeng, Pengfei Wang, Yuren Mao, Lu Chen, Xiaoze Liu, Yunjun Gao}
\IEEEauthorblockA{Zhejiang University}\IEEEauthorblockA{\{zengxc, wangpf, yuren.mao, luchen, xiaoze, gaoyj\}@zju.edu.cn
}}

\maketitle
\thispagestyle{plain}
\pagestyle{plain}

\begin{abstract}
Entity Matching (EM), which aims to identify all entity pairs referring to the same real-world entity from relational tables, is one of the most important tasks in real-world data management systems. Due to the labeling process of EM being extremely labor-intensive, unsupervised EM is more applicable than supervised EM in practical scenarios. Traditional unsupervised EM assumes that all entities come from two tables; however, it is more common to match entities from multiple tables in practical applications, that is, multi-table entity matching (multi-table EM). 
Unfortunately, effective and efficient unsupervised multi-table EM remains under-explored. To fill this gap, this paper formally studies the problem of unsupervised multi-table entity matching and proposes an effective and efficient solution, termed as \textsf{MultiEM}. \textsf{MultiEM} is a parallelable pipeline of \emph{enhanced entity representation}, \emph{table-wise hierarchical merging}, and \emph{density-based pruning}. Extensive experimental results on six real-world benchmark datasets demonstrate the superiority of \textsf{MultiEM} in terms of effectiveness and efficiency.
\end{abstract}

\begin{IEEEkeywords}
Entity Matching, Data Integration
\end{IEEEkeywords}

\section{Introduction}
\label{intro}

Entity Matching (EM), one of the most fundamental and significant tasks in data management and data preparation, aims to identify all pairs of entity records that refer to the same real-world entity from relational tables.
Most existing studies \cite{li2020deep,wang2022promptem,li2021auto,tu2022domain} assume that all entities come from two tables, namely two-table entity matching. This assumption limits the application of these methods in practical scenarios involving multiple tables. For example, some online price comparison services (e.g., Pricerunner \cite{pricerunner} and Skroutz \cite{skroutz}) compare the prices of the same product on multiple e-commerce platforms so that shoppers can search for the best deals. Because there are different titles or descriptions on different e-commerce platforms for identical products, one of the most important steps is to effectively identify the same product from multiple sources. As shown in Figure \ref{fig:example}, four entities from different sources refer to the same real-world entity (i.e., Apple iPhone 8 plus
64GB silver) with similar but different titles and colors. Furthermore, multiple sources lead to an increase in the number of entities, which imposes a higher requirement on the efficiency of entity matching.

\begin{figure}
    \centering
    \includegraphics[width=3.3in]{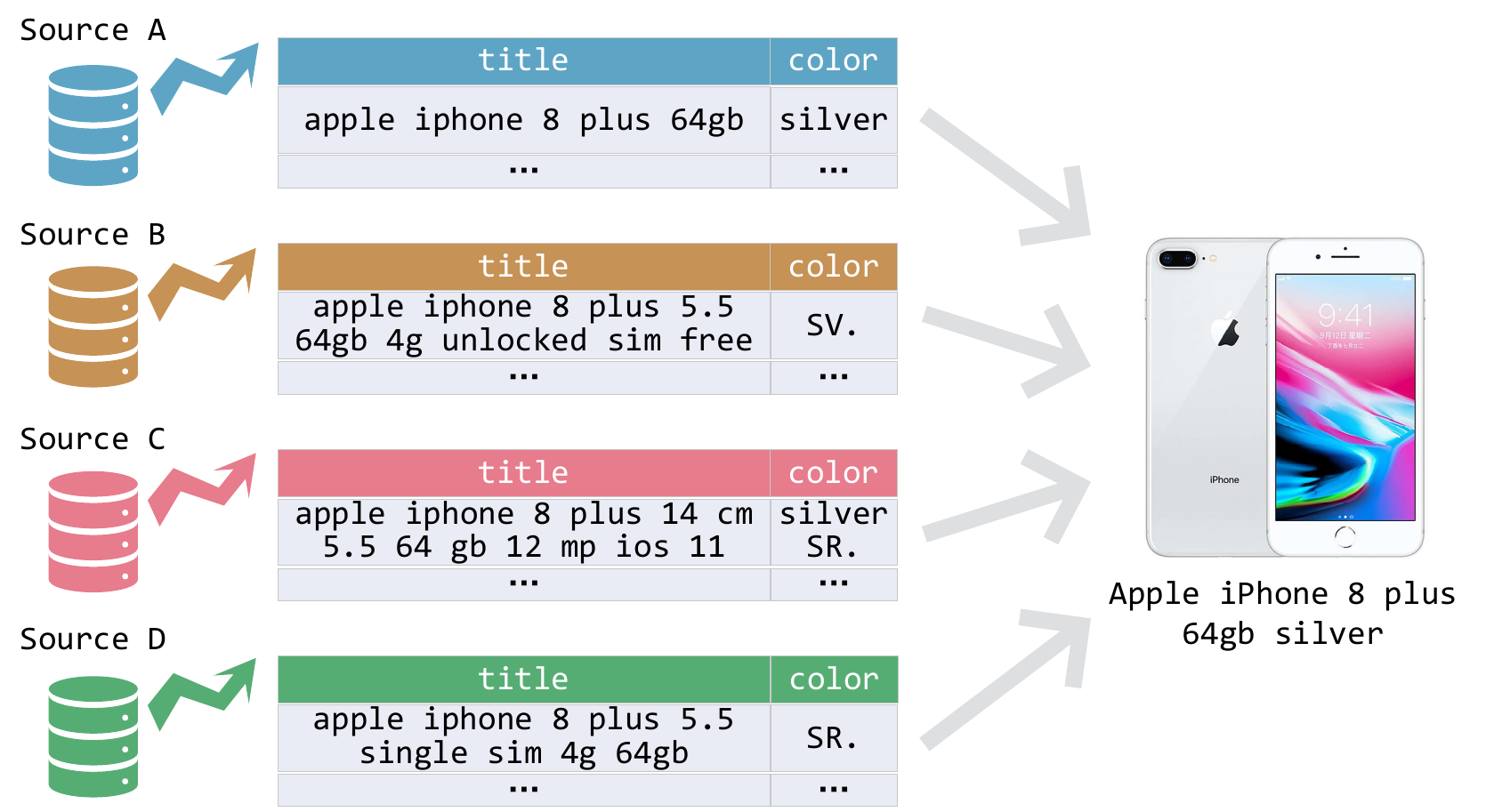}
    \caption{An example of Multi-Table Entity Matching.}
    \label{fig:example}
\end{figure}

Most existing EM methods are typically performed in a supervised~\cite{li2020deep,tu2022domain,mudgal2018deep} or semi-supervised \cite{wang2022promptem,primpeli2021graph} learning way, which rely on large amounts of labeled data, and thus is extremely labor-intensive~\cite{wu2020zeroer,li2021auto}. Therefore, performing entity matching in an unsupervised manner has become an urgent need recently. Existing unsupervised methods for multi-table EM (i.e., \textsf{MSCD-HAC} \cite{saeedi2021matching} and \textsf{MSCD-AP}~\cite{lerm2021extended}) run in a clustering way and perform poorly in terms of effectiveness and efficiency: (1) They are influenced by the complexity of clustering algorithms (i.e., hierarchical agglomerative clustering and affinity propagation) and have problematic scalability. (2) Their absence of effective entity representations poses a significant predicament, as the accuracy of clustering relies on the quality of the representations. Sophisticated analyses on the ineffectiveness and inefficiency of the existing multi-table EM methods can be found in Section \ref{experiments}.





Motivated by the above considerations, we study the problem of unsupervised multi-table entity matching. Our goal is to develop an efficient and effective solution for multi-table entity matching without the need for human-labeled data, which is a challenging endeavor. The challenges are mainly two-fold:


\noindent\textbf{Challenge \uppercase\expandafter{\romannumeral 1}:} \emph{How can multiple tables be matched efficiently?}  Recently, there has been an urgent need to efficiently match entities in large-scale data scenarios \cite{Gazzarri2022ProgressiveER,Papadakis2021TheFG}. Furthermore, in multi-table EM, multiple data sources bring a surge in the number of entities, putting forward a higher need for the matching efficiency. Existing unsupervised multi-table EM methods can be divided into three categories: clustering-based methods \cite{saeedi2021matching,lerm2021extended} and two extended methods from two-table EM \cite{li2021auto} using pairwise matching and chain matching, respectively. All of these approaches suffer from inefficiency issues.

Firstly, clustering-based methods involve clustering operations that are not inefficient. Secondly, pairwise matching-based methods (illustrated in Figure \ref{fig:comparision}(a)) are directly extended from the two-table EM methods by means of pairwise comparison between any two tables, which suffer quadratic time complexity. Besides, chain matching-based methods (illustrated in Figure \ref{fig:comparision}(c)) extend two-table EM by matching tables one by one, which is not parallelizable. Moreover, as the size of the base table increases, the two-table matching efficiency gradually declines. Overall, efficient multi-table EM methods remain less explored.


\begin{figure}
    \centering
    \includegraphics[width=3.3in]{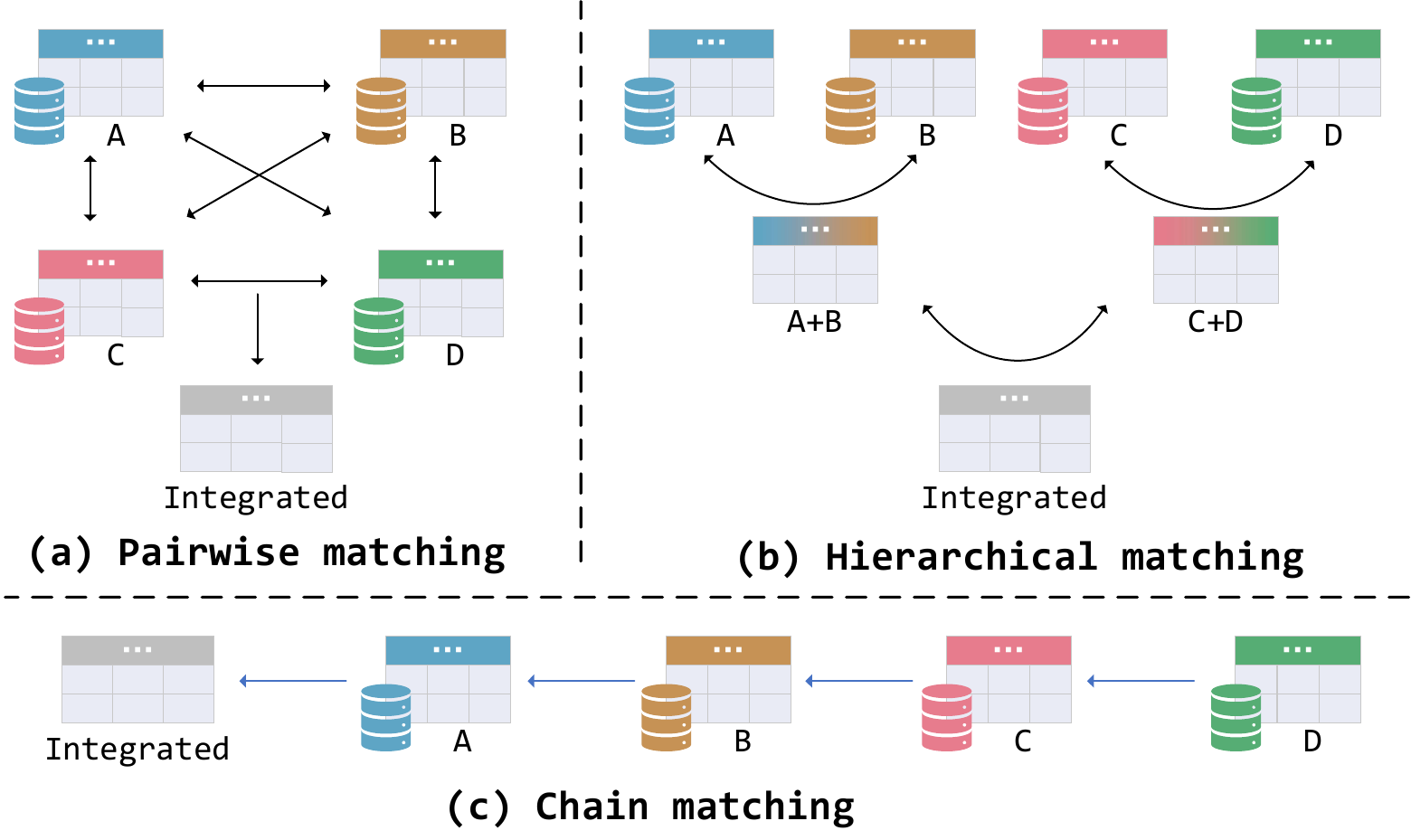}
    \caption{Different solutions for Multi-Table Entity Matching.}
    \label{fig:comparision}
\end{figure}


\noindent\textbf{Challenge \uppercase\expandafter{\romannumeral 2}:} \emph{How can multiple tables be matched effectively?} As one of the most significant tasks in data management, the effectiveness of entity matching is crucial. However, existing methods for unsupervised multi-table entity matching face two major obstacles in effectiveness. The first is the limited capability of entity representation, and the second is the existence of transitive conflicts for entity matching.

Data representation is the core for improving the effectiveness of most unsupervised data integration tasks \cite{cappuzzo2020creating,Yin2020TaBERTPF}. However, existing unsupervised entity matching methods have limitations in terms of the effectiveness of entity representation. \textsf{EMBDI} \cite{cappuzzo2020creating} learns local embeddings of entities through random walks on the heterogeneous graph, which relies more on co-occurrence relationships on a graph and neglects high-level semantic information. \textsf{AutoFJ} \cite{li2021auto}, \textsf{MSCD-HAC} \cite{saeedi2021matching}, and \textsf{MSCD-AP} \cite{lerm2021extended} use only n-gram tokenization and string-based similarity functions, which may lack some useful contextual information. Furthermore, these methods treat each attribute of the record equally without considering that each attribute may contribute differently to the representation.




Transitive conflicts is another important factor that significantly influences the performance of the effectiveness of multi-table entity matching. In multi-table EM, we need to find matched tuples (i.e., a \emph{group} of equivalent entities) rather than matched pairs. Thus, it requires aggregating matched pairs into tuples, which involves transitivity. Transitivity is a key property of entity matching, that is, if $A$ matches $B$ and $B$ matches $C$, then $A$ can be inferred to match $C$ as well. However, existing EM methods inevitably make incorrect predictions, which are propagated and result in transitive conflicts. These conflicts pose a significant obstacle to the effectiveness of matching. Moreover, as the number of tables increases, such conflicts become more complex.

To address the above two challenges, we propose an unsupervised multi-table entity matching method, dubbed \textsf{MultiEM}, which can achieve efficient and effective multi-table entity matching. In \textsf{MultiEM}, we firstly formulate multi-table EM as a two-step process (i.e., \emph{merging} and \emph{pruning}). To overcome the efficiency challenge, we present a parallelizable table-wise hierarchical merging algorithm to accelerate the matching of multiple tables. Furthermore, to address the effectiveness challenge, in MultiEM, we enhance the entity representation quality by a novel automated attribute selection strategy and handle transitive conflicts by hierarchical merging, which explicitly avoids the disjointed process of generating matched pair and converting pairs to tuples. Moreover, we develop a density-based pruning strategy to erase outliers and further improve the matching effectiveness. Our contributions are summarized as follows.



\begin{itemize}[topsep=0pt,itemsep=0pt,parsep=0pt,partopsep=0pt,leftmargin=*]
    \item \emph{Unsupervised Multi-Table EM.} To the best of our knowledge, this is the first work to formally define unsupervised multi-table entity matching problem and formulate it as a two-step (i.e., \emph{merging} and \emph{pruning}) process. 
    \item \emph{Efficient and Effective Pipeline.} We propose a novel unsupervised multi-table entity matching method, dubbed \textsf{MultiEM}, which can achieve state-of-the-art performance on efficiency and effectiveness.
    \item \emph{Extensive Experiments.} We conduct a comprehensive experimental evaluation on six real-world datasets with various domains, sizes, and numbers of sources. Extensive experimental results demonstrate the superiority of our proposed \textsf{MultiEM} in terms of effectiveness and efficiency.
\end{itemize}

\section{Preliminaries}
\label{preliminaries}

\begin{table}
\caption{Symbols and description.}
\label{table:notation}
\centering
\begin{tabular}{c|c}
\toprule
Symbol & Description \\ \midrule
    $\mathcal{D}$ & A set of tables $\mathcal{D} = \{E_1, E_2, \cdots, E_S\}$\\
    $E$ & A relational table $E=\{e_1, e_2, \cdots, e_m\}$ \\
    $S$ & The number of tables \\
    $e_i$ & An entity $e_i=\{(\text{attr}_j, \text{val}_j) \vert 1\leq j \leq p\}$\\
    $\text{attr}_j$ & An attribute name of the entity \\
    $\text{val}_j$ & A value of the entity \\
    $\mathcal{M}$  & The Sentence-BERT encoder \\
    $x$ & The text sequence of the entity \\
    $w$ & The encoded result of the entity \\
    $h$ & The embedding of the entity \\
    $n$ & The number of entities in one table\\
    \bottomrule
\end{tabular}
\end{table}

In this section, we illustrate the definition of typical two-table entity matching and then formally define the multi-table entity matching. Additionally, we provide an overview of the relevant background materials and techniques utilized in subsequent sections. Table \ref{table:notation} summarizes the symbols that are frequently used throughout this paper.

\subsection{Problem Formulation}

\begin{definition}
\label{pre:two-table}
(two-table entity matching). Given two relational tables $E_A$ and $E_B$, two-table entity matching (two-table EM) aims to identify all \textit{pairs} of records $\mathcal{P}=\{(e^{A}_{i}, e^{B}_{j})\}_{u}$, where $e^{A}_{i} \in E_A$, $e^{B}_{j} \in E_B$, that refer to the same real-world entity. 
\end{definition}

Two-table entity matching consists of two steps in sequence: \emph{blocking} and \emph{matching} \cite{li2020deep}. \emph{Blocking} is a coarse-grained step to filter out mismatched entity pairs, reducing the number of candidate pairs for matching. \emph{Matching} is a subsequent fine-grained step to determine whether each candidate pair matches exactly.

\begin{definition}
\label{pre:multi-table}
(multi-table entity matching). Given a set of relational tables $\mathcal{E} = \{E_{1}, ..., E_{S}\}$, multi-table entity matching (multi-table EM) seeks to identify all \textit{tuples} of records $\mathcal{T}=\{(e_{1}, e_{2}, ..., e_{l})\}_{u}$, where each record is from one of the $S$ tables, that refer to the same real-world entity. Specifically, the size of each tuple $l \ge 2$.
\end{definition}

Inspired by two-table EM (\emph{blocking} and \emph{matching}), we formally define the pipeline for multi-table EM, dividing it into two key steps: \emph{merging} and \emph{pruning}. \emph{Merging} focuses on identifying potentially matched tuples across tables, while \emph{pruning} aims to determine the most accurate matches among the candidates.

Note that there is a significant difference between two-table and multi-table EM. Two-table EM aims to find all matched entity pairs. However, multi-table EM intends to identify matched tuples, which refer to a group of equivalent entities found across multiple tables. As analyzed in Section \ref{intro}, multi-table EM is more practical in the real world, with huge challenges in terms of efficiency and effectiveness.

\subsection{Sentence-BERT}

Sentence-BERT \cite{reimers2019sentence} is a variant of BERT model based on Siamese and triplet network structures. Sentence-BERT is appropriate for sentence representations and can be used for anything serialized into sentences \cite{ge2021collaborem, Wang2021MachampAG}. As a result, structural entities can be serialized into sentences based on specific rules and then converted into embeddings using Sentence-BERT.

\noindent\textbf{Serialization.} Since pre-trained language models (e.g., Sentence-BERT \cite{reimers2019sentence}) take sentences as input, we adapt them to the EM task by serializing each entity into a text sequence. We omit attribute names of the entity and concatenate all attribute values to get a text sequence. Specifically, for each entity $e=\{(\text{attr}_j, \text{val}_j) \vert 1\leq j \leq p\}$, it can be serialized as follows:


$$
serialize(e)::= \operatorname{ val_1 \, val_2 \, \cdots \, val_{p-1} \, val_p}
$$

As an example in Figure \ref{fig:example}, the entity A1 can be serialized as "apple iphone 8 plus 64gb silver".

\noindent\textbf{Representation.} Formally, given a Sentence-BERT model $\mathcal{M}$ and an input text sequence $x = \{t_1, t_2, \cdots, t_u\}$. First, apply a tokenizer to encode $x$ and feed the encoded result $w = \{v_1, v_2, \cdots, v_u\}$ to the model $\mathcal{M}$. Then a pooling method is applied for the embeddings of each token to obtain a fixed length embedding $h=pooling(\mathcal{M}(w))$ of the entity.

\subsection{Approximate Nearest Neighbor Search (ANNS)}

Nearest Neighbor Search, which aims at finding the top-k nearest objects to the query object in a reference set, is a crucial operation in various applications such as databases, computer vision, multimedia, and recommendation systems \cite{li2019approximate}. However, finding the exact nearest neighbor in high-dimensional space is generally computationally expensive. As a result, many researchers have focused on developing Approximate Nearest Neighbor Search (ANNS), which only returns sufficiently nearby objects. That is useful and efficient for several practical problems.

There are many different types of competitive methods for ANNS, such as LSH-based methods (e.g., QALSH \cite{Huang2015QueryAwareLH}), encoding-based methods (e.g., SGH \cite{Jiang2015ScalableGH}), tree-based methods (e.g., FLANN \cite{Muja2014ScalableNN}), and neighborhood-based methods (e.g., HNSW \cite{malkov2020efficient}). These methods are implemented in different ways with advantages and suitable for different scenarios.
\vspace{5mm}
\section{Method}
\label{method}

\begin{figure*}
\centering
\includegraphics[width=7in]{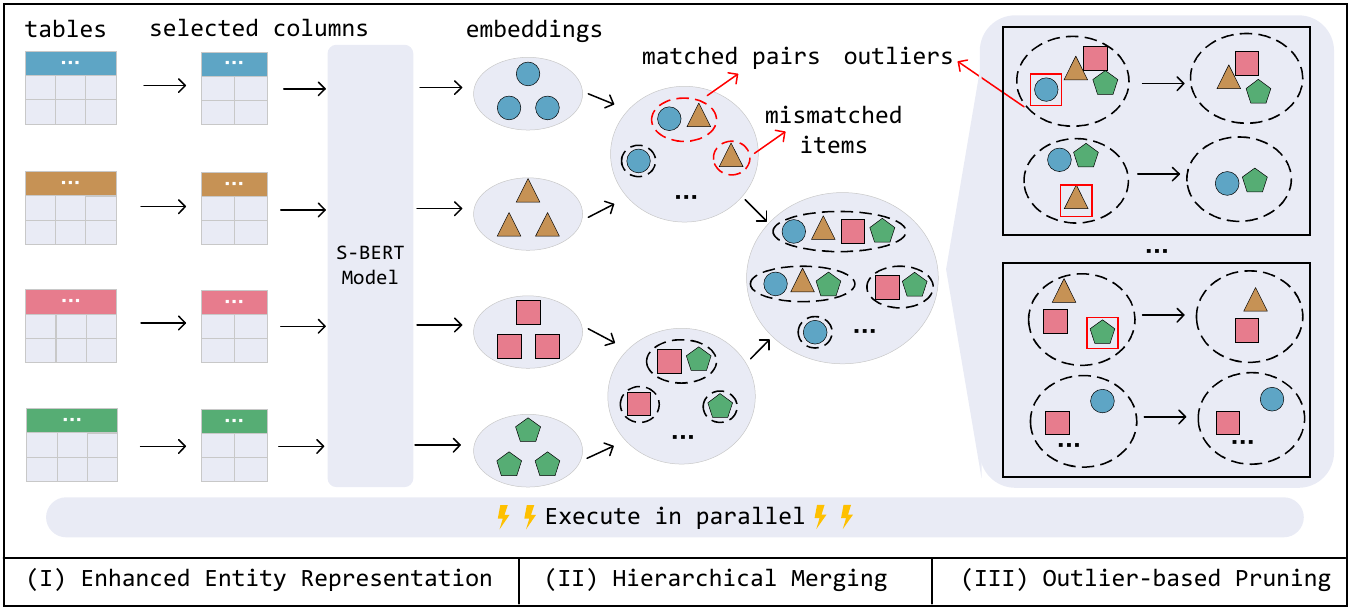}
\caption{The proposed \textsf{MultiEM} framework.}
\label{fig:framework}
\vspace{-3mm}
\end{figure*}

In this section, we present a highly
efficient and effective approach for multi-table entity matching, dubbed \textsf{MultiEM}. We first introduce the overall framework, followed by details of three modules: \emph{Enhanced Entity Representation}, \emph{Table-wise Hierarchical Merging}, and \emph{Outlier-based Pruning}. Finally, we emphasize the high parallelizability of \textsf{MultiEM} and present its parallelized version, namely \textsf{MultiEM}(parallel).

\subsection{Overview of MultiEM}
\label{method:overview}

As illustrated in Figure \ref{fig:framework}, we sequentially solve multi-table EM in three phases, i.e., \emph{representation}, \emph{merging}, and \emph{pruning}.  In the first (representation) step, all entities are serialized and converted into high-quality embeddings based on automated attribute selection. And then, in the second (merging) step, we propose a table-wise hierarchical merging algorithm to generate candidate tuples efficiently. In the last (pruning) step, we design a pruning strategy for each candidate tuple to further improve matching performance. Furthermore, \textsf{MultiEM} has a highly parallelizable design. In the merging phase, the algorithm can merge all table pairs independently. Similarly, each tuple can be pruned independently in the pruning phase without sacrificing matching performance.

\subsection{Enhanced Entity Representation}
\label{method:representation}

The quality of representations significantly impacts the effectiveness of downstream tasks, as supported by multiple studies \cite{deng2022turl,yin2020tabert,yamada2020luke}. It is especially true in unsupervised entity matching scenarios since no matched/mismatched labels exist. As mentioned before, Sentence-BERT \cite{reimers2019sentence} has demonstrated its power in sentence semantic representation, which can support many downstream tasks effectively without fine-tuning, such as retrieval and query \cite{kim2022reweighting,arabzadeh2021matches,lim2022q2r}. Therefore, we use a pre-trained Sentence-BERT \cite{reimers2019sentence} model to represent all entities without additional training costs to keep the lightweight and high efficiency of \textsf{MultiEM}, which will be analyzed in Section \ref{exp:efficiency}. However, this way may not be good enough as it considers all attributes of entities, regardless of their relevance to entity matching. Intuitively, some attributes may have no or even negative impacts on the Sentence-BERT representations.

\begin{table}
\setlength{\tabcolsep}{8pt}
\caption{Entity $e_a$, $e_b$ and $e_c$.}
\label{rep:enttiy_abc}
\begin{tabular}{c|c|c|c|c}
\toprule
 & \textbf{id} & \textbf{title} & \textbf{artist} & \textbf{album} \\ \midrule
$e_a$ & \underline{WoM14513028} & Megna's & Tim O’Brien & \textit{Chameleon} \\ \midrule
$e_b$ & \underline{WoM94369364} & Megna's & Tim O’Brien & Chameleon \\ \midrule
$e_c$ & WoM14513028 & Megna's & Tim O’Brien & \textit{The Hitmen} \\ \bottomrule
\end{tabular}
\end{table}

\begin{example} 
\label{example:entity_representation}
As illustrated in Table \ref{rep:enttiy_abc}, given one structural entity $e_a$ and replace its attribute \emph{id} and \emph{album} respectively to get two entities $e_b$ and $e_c$. And then, they are represented by the pre-trained Sentence-BERT model. It is observed that the cosine similarity of $e_a$ and $e_b$ is 0.91, and that of $e_a$ and $e_c$ is 0.79. In other words, changes made to the \emph{id} do not significantly impact the entity's embedding. This finding suggests that some attributes may not be understood well by Sentence-BERT and could potentially have a negative effect.
\end{example}


Based on this intuition, we design a general module based on automated attribute selection to enhance the entity representation. Some studies \cite{hall2003benchmarking, Paulsen2023SparklyAS} use information entropy or TF-IDF scores to measure the importance of attributes. Nevertheless, these metrics do not apply to our method, as they are based on word/phrase frequency, which differs from our objective of enhancing SentenceBERT-based representation. Example \ref{example:entity_representation} demonstrates that replacing the value of a significant attribute results in a larger change in the embedding than replacing an insignificant attribute. Leveraging this insight, we propose an algorithm to select significant and valuable attributes, including the following key steps:

\begin{enumerate}
    \item Select an attribute and shuffle the values of all entities;
    \item Generate the new embeddings with new values;
    \item Compute the distance of the new and old embeddings for each entity;
    \item Average all entities' distance as the significance score;
    \item Repeat steps 1-4 to compute significance scores for all attributes;
    \item Select more significant attributes based on a threshold $\gamma$.
\end{enumerate}

The pseudo code is shown in Algorithm \ref{algo:attribute-selection}. We optimize the raw algorithm based on random sampling (Line 2) to reduce the time overhead, as a subset of entities (with ratio $r$) is sufficient to calculate the significance scores for large-scale datasets.

\begin{algorithm}[t]
\caption{Automated Attribute Selection}
\label{algo:attribute-selection}
\LinesNumbered
\DontPrintSemicolon
    \KwIn{a set of tables $\mathcal{D}=\{E_1, E_2, \cdots, E_S\}$ with the same schema, a Sentence-BERT model $M$, hyperparameters $r$, $\gamma$}
    \KwOut{a set of selected attributes $selectedAttrs$}
    \tcp{Concatenate all tables into one table.}
    $E \gets \operatorname{concat}(E_1, E_2, \cdots, E_S$) \;
    \tcp{Sample some rows of the table.}
    $E \gets \operatorname{sample}(E)$ \;
    \tcp{Generate the initial embeddings.}
    $H \gets M(E)$ \;
    $selectedAttrs \gets []$ \;
    \tcp{Calculate the significance score of each attribute.}
    \For{$attr \in \operatorname{attributes}(E)$}{
        $E^{\prime} \gets E$ \;
        \tcp{Shuffle the values of this attribute.}
        $E^{\prime}[attr] \gets \operatorname{shuffle}(E^{\prime}[attr])$ \;
        \tcp{Generate the new embeddings.}
        $H^{\prime} \gets M(E^{\prime})$ \;
        \tcp{Calculate the mean similarity.}
        $sim \gets  \operatorname{distance}(H, H^{\prime})$\;
        \If{$sim \geq \gamma$}{
            $selectedAttrs \gets \operatorname{append}(selectedAttrs, attr)$
        }
    }
\Return{$selectedAttrs$}
\end{algorithm}

\subsection{Table-wise Hierarchical Merging}
\label{method:merging}
As mentioned in Section \ref{intro}, existing two-table EM methods \cite{li2020deep,wang2022promptem,li2021auto,tu2022domain} need to be extended to match multiple tables by pairwise matching (i.e., Figure \ref{fig:comparision}(a)) or chain matching (i.e., Figure \ref{fig:comparision}(c)). However, both approaches suffer from inefficiencies with high computational complexity (i.e., $T_{p}(S,n) \geq O(S^2 2kn\log n)$ and $T_{c}(S,n) \geq O(S^2 kn\log n)$). In addition, they need to generate all matched pairs first and then combine pairs to tuples, which is disjointed and hampered by transitive conflicts, thus affecting effectiveness. To address these issues, we propose a table-wise hierarchical merging algorithm (i.e., Figure \ref{fig:comparision}(b)) with lower time complexity (i.e., $T(S,n)=O(Skn\log S \log n)$) and can explicitly avoid the disjointed process described above. Specifically, as described in Algorithm \ref{algo:hierarchical_merging}, every two tables are merged into a single table (Line 4) hierarchically and iteratively until one table remains (Line 7) as the final result. However, how to deal with the merging of given two tables to ensure the effectiveness of matching is not trivial.


To this end, we elaborately design an ANNS-based two-table merging strategy to find some candidate tuples with its pseudo-code in Algorithm \ref{algo:merging_two_tables}. The core of this strategy is to merge the matched entities and keep the mismatched ones in the next hierarchy. It contains two steps as follows.

In the first step, we leverage HNSW \cite{malkov2020efficient}, an ANN index based on the navigable small world graphs, to balance the accuracy and efficiency. We build the indexes on every two tables and employ them to find all mutual top-K items with a distance less than $m$ as matched entity pairs $\mathcal{P}_m$  (Lines 3-5).

\begin{equation}
\label{equ:find_mutual_topk}
\mathcal{P}_{m}=\{(e,e^{\prime}) \vert e \in \operatorname{topK}(e^{\prime}) \wedge e^{\prime} \in \operatorname{topK}(e) \wedge \operatorname{dist}(e,e^{\prime}) \leq m\}
\end{equation}
Here, $e$ comes from $E_i$, $e^{\prime}$ is from $E_j$, and $\operatorname{dist}$ represents the distance function.

In the second step, we merge all the matched entity pairs based on the transitivity and retain the mismatched ones into a new table $E_{mer}$ (Lines 6-10).




\begin{algorithm}[t]
\caption{Table-wise Hierarchical Merging}
\label{algo:hierarchical_merging}
\LinesNumbered
\DontPrintSemicolon
    \KwIn{a set of tables $\mathcal{D}=\{E_1, E_2, \cdots, E_S\}$}
    \KwOut{an integrated table $E_{inte}$}
    \tcp{Iterative merging until one table remains.}
    \While {$len(\mathcal{D})>1$}{
        $\mathcal{D}_{temp} \gets$ empty list \;
        \tcp{Randomly sample two tables repeatedly.}
        \While {$E_i, E_j \gets randomSample(\mathcal{D})$}{
            \tcp{Apply the two-table merging strategy.}
            $E_{ij} \gets \operatorname{merging}(E_i, E_j)$ \;
            $\mathcal{D}_{temp} \gets \operatorname{append}(\mathcal{D}_{temp}, E_{ij})$ \;
        }
        $\mathcal{D} \gets \mathcal{D}_{temp}$ \;
    }
    $E_{inte} \gets \mathcal{D}[0]$ \;
\Return{$E_{inte}$}
\end{algorithm}


\begin{algorithm}[t]
\caption{Two-table Merging Strategy}
\label{algo:merging_two_tables}
\LinesNumbered
\DontPrintSemicolon
    \KwIn{two tables $E_i$ and $E_j$; the query hyperparameters $k$, $m$}
    \KwOut{one merged table $E_{mer}$}
    \tcp{Generate embeddings of each item.}
    $H_i \gets \operatorname{Representation}(E_i)$ \;
    $H_j \gets \operatorname{Representation}(E_j)$ \;
    \tcp{Find mutual top-K pairs by ANNS.}
    $\mathcal{P}_{ij} \gets \operatorname{ANNS}(H_i, H_j, k, m)$ \;
    $\mathcal{P}_{ji} \gets \operatorname{ANNS}(H_j, H_i, k, m)$ \;
    $\mathcal{P}_m \gets \mathcal{P}_{ij} \cap \mathcal{P}_{ji}$ \;
    \tcp{Get matched pairs of each single table.}
    $\mathcal{P}_i \gets  \operatorname{MatchedPairs}(E_i)$ \;
    $\mathcal{P}_j \gets  \operatorname{MatchedPairs}(E_j)$ \;
    \tcp{Merge based on the transitivity.}
    $P_{matched} \gets \operatorname{Merge}(\mathcal{P}_{m}, \mathcal{P}_i, \mathcal{P}_j)$ \;
    \tcp{Generate a new table.}
    $E_{mismatched} \gets \{x \vert x \in E_{i} \cup E_{j} \wedge x \notin \mathcal{P}_{matched} \}$ \;
    $E_{mer} \gets \mathcal{P}_{matched} \cup E_{mismatched}$ \;
\Return{$E_{mer}$}
\end{algorithm}

We analyzed the time complexity to demonstrate the theoretical superiority of the proposed hierarchical merging approach over pairwise matching and chain matching in efficiency.

Given $S$ tables with average size $n$. The complexities of pairwise matching, chain matching, and our proposed hierarchical merging are as follows:

\begin{lemma} Denote the time complexity of pairwise matching as  $T_{p}(S,n)$, we have 
\label{lemma:pairwise_matching}
\begin{equation}
    T_{p}(S,n) \geq O(S^2 2kn\log n).
\end{equation}
\end{lemma}

\begin{proof}
For pairwise matching of $S$ tables, $\tbinom{S}{2}$ times of two-table EM methods are applied. Therefore, its complexity depends on the complexity of the applied two-table EM method, denoted as:

\begin{equation}
    T_{p}(S,n)=O(S^2 f(n))
\end{equation}
Here, $f(n)$ is the complexity for matching two tables.

Suppose that the mutual top-K search (i.e., with complexity $O(2kn\log n)$) is applied to match two tables. Therefore, the overall complexity of pairwise matching is computed as:

\begin{equation}
    T_{p}(S,n)=O(S^2 2kn\log n)
\end{equation}

For other more complex EM methods \cite{li2020deep,wang2022promptem,li2021auto,tu2022domain}, $f(n)$ is much higher than $O(2kn\log n)$, so the overall complexity: 

\begin{equation}
    T_{p}(S,n) \geq O(S^2 2kn\log n)
\end{equation}
\end{proof}

\begin{lemma}Denote the time complexity of chain matching as $T_{c}(S,n)$, we have 
\label{lemma:chain_matching}
\begin{equation}
    T_{c}(S,n) \geq O(S^2 kn\log n).
\end{equation}
\end{lemma}
\begin{proof}
    For chain matching of $S$ tables, first, the base table is selected, and then the other $S-1$ tables are matched one by one. We refer the above $f(n)$ as the matching complexity of two tables. Here, $f(n)=O(kn\log n^{\prime}+kn^{\prime}\log n)$ because the sizes of the two tables are different. As matching, the unmatched entities are retained, leading to an increase in the size of the base table. Therefore, the overall complexity:

\begin{equation}
    \begin{split}
        T_c(S,n)& = \sum_{i=1}^{S-1} O(kin\log n+kn\log in) \\ 
        &=\sum_{i=1}^{S-1} O(kin\log n) + \sum_{i=1}^{S-1} O(kn\log in) \\ 
        &=O(kn(\sum_{i=1}^{S-1}i \log n+\sum_{i=1}^{S-1}\log n+\sum_{i=1}^{S-1}\log i)) \\
        &=O(S^2kn\log n + Skn\log n + kn\sum_{i=1}^{S-1}\log i) \\
        &\geq O(S^2kn\log n)
    \end{split}
\end{equation}
\end{proof}

\begin{lemma} Denote the time complexity of  hierarchical merging as $T_{c}(S,n)$, we have 
\label{lemma:hierarchical_merging}
\begin{equation}
    T(S,n)=O(Skn\log S \log n).
\end{equation}
\end{lemma}

\begin{proof}
    For each hierarchy $i$ from 1 to $\log S$ with $\frac{S}{2^{i-1}}$ tables, we apply the two-table merging function (i.e., Algorithm \ref{algo:merging_two_tables}) to every two tables. Therefore, the time complexity can be expressed as:

\begin{equation}
    T(S,n)=\sum_{i=1}^{\log S} \frac{S}{2^{i}} t(i)
\end{equation}
Here, $t(i)$ denotes the complexity of merging two tables at hierarchy $i$, that is, $t(i)=O(2kn^{\prime}\log n^{\prime})$, where $n'$ is the size of the tables at this hierarchy.

To be more specific, for two tables of size $n$, the size of the merged table $n^{\prime}<=2n$. In conclusion, the final time complexity can be calculated as follows:

\begin{equation}
    \begin{split}
        T(S,n)& \leq \sum_{i=1}^{\log S} \frac{S}{2^{i}} O(2k2^{i-1}n\log (2^{i-1}n))  \\
        & \leq O(Skn \sum_{i=1}^{\log S} \log (2^{i-1}n)) \\
        & \leq O(Skn (\sum_{i=1}^{\log S} \log 2^{i-1} + \sum_{i=1}^{\log S} \log n)) \\
        & \leq O(Skn (\log S \frac{\log S -1}{2} +\log S \log n)) \\
        & \leq O(Skn \log S (\frac{\log S -1}{2}+\log n))
    \end{split}
\end{equation}

Since $S \ll n$ in almost all cases, the complexity can be expressed as $O(Skn \log S \log n)$
\end{proof}

Overall, we demonstrate the efficiency and effectiveness of this hierarchical merging algorithm in two aspects. On the one hand, the theoretical time complexity of the hierarchical merging algorithm is $O(Skn \log S \log n)$, which is better than pairwise matching and chain matching. On the other hand, it is both effective and efficient in experiments, which is to be evaluated in Section \ref{experiments}.

\subsection{Density-based Pruning}
\label{method:pruning}

\begin{figure}
    \centering
    \includegraphics[width=3.3in]{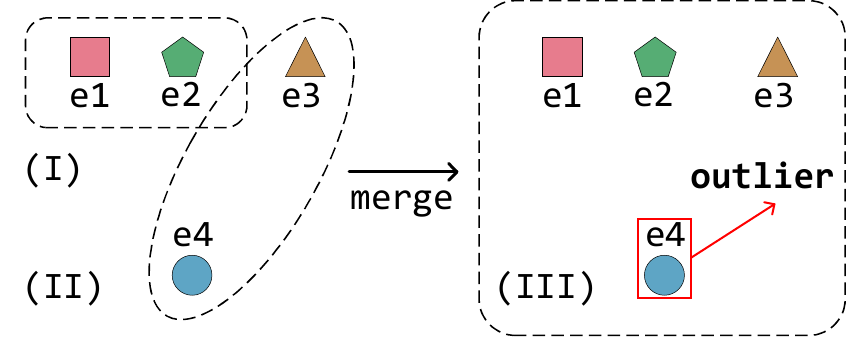}
    \caption{An intuitive example of pruning.}
    \label{fig:pruning_case}
    \vspace{-3mm}
\end{figure}

The hierarchical merging phase produces some prediction tuples in the final merged table. Nevertheless, these results are still noisy due to the locality limitations of merging. In other words, it is caused by only considering the two tables currently being merged. As shown in Figure \ref{fig:pruning_case}, first, the entities $e1$ and $e2$, $e3$ and $e4$ are merged, respectively (i.e., (\uppercase\expandafter{\romannumeral 1}) and (\uppercase\expandafter{\romannumeral 2})). Then these two pairs continue to be merged (i.e., (\uppercase\expandafter{\romannumeral 3})). However, at this point, $e4$ becomes an outlier entity in the data item $(e1, e2, e3, e4)$.


As mentioned above, we define the pruning phase as the problem of outlier detection and removal of each merging tuple. We adopt the idea of density-based \cite{Ester1996ADA,kriegel2011density} and design a density-based pruning strategy by identifying entities with different densities to improve the matching performance. Specifically, for each data item $x=\{e_1, e_2, \cdots, e_u\}$ that contains multiple entities, we first define three types of entity (i.e., \emph{core entity}, \emph{reachable entity}, and \emph{outlier entity}) as follows. 


\begin{definition}
(Core Entity). Given a data item $x=\{e_1, e_2, \cdots ,e_u\}$, an entity in it is a core entity when the indicated function $f_c(e)$ is true. Specifically, $f_c(e)$ is calculated as follows.
\end{definition}

\begin{equation}
f_c(e)=\lvert N_{\epsilon}(e,x) \rvert \geq MinPts
\end{equation}

\begin{equation}
\label{equ:dbscan_eps_core}
N_{\epsilon}(e,x)= \{e^{\prime} \vert e^{\prime} \in x \wedge \operatorname{distance}(e,e^{\prime}) \leq \epsilon \}
\end{equation}
Here, $N_{\epsilon}(e,x)$ represents the $\epsilon$-neighbor entities of $e$ in data item $x$, and $MinPts$ denotes the number of neighbors required for $e$ to become a core entity.

\begin{definition}
(Reachable Entity). A reachable entity is a non-core entity that can be reached through the core entities within its $\epsilon$-neighborhood. The formal definition of its indicator function $f_r(e)$ is as follows.
\end{definition}

\begin{equation}
f_r(e)=\lvert N_{c,\epsilon}(e,x) \rvert \geq 1
\end{equation}

\begin{equation}
\label{equ:dbscan_eps_reachable}
N_{c,\epsilon}(e,x)= \{e^{\prime} \vert e^{\prime} \in x \wedge \operatorname{distance}(e,e^{\prime}) \leq \epsilon \wedge f_{c}(e^{\prime}) \}
\end{equation}
Here, $N_{c,\epsilon}(e,x)$ denotes the core entities within the $\epsilon$-neighborhood of $e$ in data item $x$.

\begin{definition}
(Outlier Entity). An outlier entity is an entity that is neither a core entity nor a reachable entity.
\end{definition}

Next, we prune each item according to the above definitions. First, we find out the core entities (Lines 3-4), reachable entities (Lines 9-10), and outlier entities (Lines 11-12) of each item, which is described with its pseudo-code in Algorithm \ref{algo:outlier_detection}. After detecting these three kinds of entities, we remove the outlier entities in each data item and merge the other two kinds of entities (i.e., core entities and reachable entities) into a new data item. Therefore, this pruning phase can remove some errors in the merging predictions and make the results of hierarchical merging more effective, which will be evaluated in Section \ref{exp:ablation}.

Note that the pruning of each data item is independent and can be easily performed in parallel to improve efficiency. We will introduce it in detail in Section \ref{method:parallel}.

\begin{algorithm}[t]
\caption{Entity Classification for Pruning}
\label{algo:outlier_detection}
\LinesNumbered
\DontPrintSemicolon
    \KwIn{a data item $x=\{e_1, e_2, \cdots, e_u\}$; density parameters $\epsilon$ and $MinPts$}
    \KwOut{core entities $E_{c}$; reachable entities $E_{r}$; outlier entities $E_{o}$}
    \For{$e \in x$}{
        $N_{\epsilon} \gets \operatorname{Neighbors}(x,e,\epsilon)$ \;
        \If{$\lvert N_{\epsilon} \rvert \geq MinPts$}{
        $E_{c} \gets \operatorname{Append}(E_{c}, e)$ \;
        }
    }
    \For{$e \in x$}{
        $N_{\epsilon} \gets \operatorname{Neighbors}(x,e,\epsilon)$ \;
        \If{$\lvert N_{\epsilon} \rvert \geq MinPts$}{
        \Continue
        }
        \ElseIf{$N_{\epsilon} \cap E_c$}{
        $E_{r} \gets \operatorname{Append}(E_{r}, e)$ \;
        }
        \Else{
        $E_{o} \gets \operatorname{Append}(E_{o}, e)$ \;
        }
    }
\Return{$E_{c}$, $E_{r}$ and $E_{o}$}
\end{algorithm}

\subsection{MultiEM in Parallel}
\label{method:parallel}

The design of \textsf{MultiEM} enables it to be extended to the parallel mode to further boost efficiency without compromising the matching performance. Specifically, in the merging phase, each pair of tables in every hierarchy is independent and can be merged in parallel. Moreover, we apply a parallel extension in the pruning phase by partitioning tuples.

\noindent\textbf{Merging in parallel.} To perform merging in parallel, all table pairs are divided into multiple groups and assigned to different computing cores. Once the calculation of the current hierarchy is completed, the merged tables are aggregated and prepared for the subsequent merging.

\noindent\textbf{Pruning in parallel.} Similarly, in the pruning phase, each data item's pruning is independent and can be executed in parallel for greater efficiency. To achieve this, the merging predictions can be divided into multiple parts and assigned to different computational cores.
\section{Experiments}
\label{experiments}

In this section, we present an experimental evaluation of \textsf{MultiEM}, using six real-world datasets. Our evaluation aims to answer the following research questions:

\begin{itemize}[topsep=0pt,itemsep=0pt,parsep=0pt,partopsep=0pt,leftmargin=*]
    \item \textbf{RQ1}: How does \textsf{MultiEM} compare to state-of-the-art methods in matching effectiveness?
    \item \textbf{RQ2}: How efficient is \textsf{MultiEM} in terms of time and memory usage?
    \item \textbf{RQ3}: What is the influence of each key module on the effectiveness and efficiency of \textsf{MultiEM}?
    \item \textbf{RQ4}: How do different hyperparameters affect the performance of \textsf{MultiEM}?
\end{itemize}

\subsection{Experimental Setup}
\label{exp:setup}

\noindent\textbf{Datasets.} We use six public real-world datasets with various domains, sizes, and numbers of sources.
The statistics of the datasets are summarized in Table \ref{table:datasets}. The dataset Shopee comes from \cite{shopee-product-matching}, and the other five datasets are from \cite{saeedi2021matching}.

\begin{table}
\renewcommand\arraystretch{1.25}
\setlength{\tabcolsep}{4pt}
\centering
\caption{Statistics of the datasets used in our experiments.}
\label{table:datasets}
\begin{threeparttable}
\begin{tabular}{c|cccccc}
\toprule
\textbf{Name} & \textbf{Domain} & \textbf{Srcs} & \textbf{Attrs} & \textbf{Entities} & \textbf{Tuples} & \textbf{Pairs} \\ \hline
\textbf{Geo} & geography & 4 & 3 & 3,054 & 820 & 4,391 \\
\textbf{Music-20} & music & 5 & 5 & 19,375 & 5,000 & 16,250 \\
\textbf{Music-200} & music & 5 & 5 & 193,750 & 50,000 & 162,500 \\
\textbf{Music-2000} & music & 5 & 5 & 1,937,500 & 500,000  & 1,625,000\\
\textbf{Person} & person & 5 & 4 & 5,000,000 & 500,000 & 3,331,384 \\
\textbf{Shopee} & product & 20 & 1 & 32,563 & 10,962 & 54,488 \\ \bottomrule
\end{tabular}
\begin{tablenotes}
\footnotesize
\item[1] ``Tuples'' denotes the number of matched tuples; ``Pairs'' represents the number of matched pairs.
\item[2] ``Srcs'' means the number of sources, that is, the number of tables $S$ described in Section \ref{pre:multi-table}. For example, ``4'' denotes that there are four tables in the Geo dataset.
\end{tablenotes}
\end{threeparttable}
\end{table}

\noindent\textbf{Baselines.} We compare \textsf{MultiEM} with five baselines, including supervised and semi-supervised methods for \emph{two-table entity matching} (i.e., \textsf{Ditto} and \textsf{PromptEM}), a SOTA unsupervised approach for \emph{two-table entity matching} (i.e., \textsf{AutoFuzzyJoin}), and methods designed for \emph{multi-table entity matching} (i.e., \textsf{ALMSER-GB} and \textsf{MSCD-HAC}). Note that for two-table EM methods, we apply both pairwise matching and chain matching for them. And then evaluate them in the multi-table EM settings following Algorithm \ref{algo:evaluation_two_table}.

\begin{itemize}[topsep=0pt,itemsep=0pt,parsep=0pt,partopsep=0pt,leftmargin=*]
    \item \textsf{PromptEM} \cite{wang2022promptem} is a prompt-tuning based approach for low-resource generalized entity matching.
    \item \textsf{Ditto} \cite{li2020deep} is a supervised EM approach that fine-tunes a pre-trained language model with labeled data.
    \item \textsf{AutoFuzzyJoin} \cite{li2021auto} is an unsupervised fuzzy join framework that can be used for two-table entity matching.
    \item \textsf{ALMSER-GB} \cite{primpeli2021graph} is a graph-boosted active learning method for
multi-source entity resolution.
    \item \textsf{MSCD-HAC} \cite{saeedi2021matching} is an extended hierarchical agglomerative clustering algorithm for clustering entities from multiple sources.
\end{itemize}

\noindent\textbf{Implementation details.} We implement \textsf{MultiEM} in Python, the Sentence-Transformers\footnote{\url{https://www.sbert.net}} library, the hnswlib\footnote{\url{https://github.com/nmslib/hnswlib}} library, and the scikit-learn\footnote{\url{https://github.com/scikit-learn/scikit-learn}} library. We use all-MiniLM-L12-v2\footnote{\url{https://huggingface.co/sentence-transformers/all-MiniLM-L12-v2}} with mean-pooling as the backbone structure of Sentence-BERT in all our experiments. It is trained using more than 1 billion sentence pairs from multiple datasets and maps a sentence to a 384-dim dense vector. We use HNSW algorithm \cite{malkov2020efficient} in the merging phase. We follow the efficient implementation of DBSCAN~\cite{Ester1996ADA} in scikit-learn library\footnote{\url{https://scikit-learn.org/stable/modules/generated/sklearn.cluster.DBSCAN.html}} for the pruning phase. For the parallel extension, we use the Joblib\footnote{\url{https://github.com/joblib/joblib}} as the underlying parallel framework. In all our experiments, the maximum sequence length is set to 64; $k$ is set to 1; $MinPts$ is set to 2; $r$ is set to 0.05 for the large dataset with more than 5 million entities (i.e., Person) and set to 0.2 for other datasets. We tune other hyper-parameters by doing a grid search and selecting the one with the best performance. Specifically, $\epsilon$ is selected from \{0.8, 1.0\}, and $m$ is selected from \{0.05, 0.2, 0.35, 0.5\}, $\gamma$ is selected from \{0.8, 0.9\}. We use the cosine distance as the metric in the merging phase and use the euclidean distance in the pruning phase. All the experiments are conducted on a machine with an Intel Xeon Silver 4216 CPU, an NVIDIA A100 GPU, and 500GB memory. The code and all datasets are available at \url{https://github.com/ZJU-DAILY/MultiEM}. We implement each baseline as follows. 

\begin{itemize}[leftmargin=*]
    \item \textsf{PromptEM} \cite{wang2022promptem}: We implement this approach according to the original paper and public code\footnote{\url{https://github.com/ZJU-DAILY/PromptEM}}.
    \item \textsf{Ditto} \cite{li2020deep}: We implement this method according to the original paper and public code\footnote{\url{https://github.com/megagonlabs/ditto}}.
    \item \textsf{AutoFJ} \cite{li2021auto}: We implement this method following the origin paper and public code\footnote{\url{https://github.com/chu-data-lab/AutomaticFuzzyJoin}}.
    \item \textsf{ALMSER-GB} \cite{primpeli2021graph}: We implement this method according to the origin paper and public code\footnote{\url{https://github.com/wbsg-uni-mannheim/ALMSER-GB}}.
    \item \textsf{MSCD-HAC} \cite{saeedi2021matching}: We implement this method described in the original paper.
\end{itemize}

\noindent\textbf{Evaluation metrics.} Following most related studies \cite{wang2022promptem,li2020deep,tu2022domain}, we use precision (P), recall (R), and F1-score (F1) as the primary metrics. Note that in our evaluation, a prediction tuple is considered correct only if it matches the truth tuple exactly. Since most baseline methods use entity pair as the evaluation unit, for a fair comparison, we use the F1-score for pairwise matching (pair-F1) as an auxiliary metric to evaluate the matching performance for another aspect.

\begin{example}
Given a truth tuple $t=(1, 2, 3)$, while a prediction tuple $p=(1, 2, 4)$. When evaluated with F1, it is a wrong prediction. Nevertheless, when evaluated with pair-F1, tuples $t$ and $p$ are parsed into pairs $\{(1, 2), (1, 3), (2, 3)\}$ and $\{(1, 2), (1, 4), (2, 4)\}$ respectively. Since the $(1, 2)$ is a truth pair, the precision and recall are both $\frac{1}{3}$, and the pair-F1 score is calculated as $\frac{1}{3}$. In general, F1-score is a strict metric, while pair-F1 is looser.
\end{example}

For supervised/semi-supervised methods (i.e., \textsf{PromptEM}, \textsf{Ditto}, and \textsf{ALMSER-GB}) that require training samples, we randomly sample 5\% of the ground truth as the train set and 5\% as the valid set. For the test set, we use the entire ground truth and randomly sample $P$ mismatched pairs for each pair for comprehensive evaluation. $P$ is set to 100 for small datasets (i.e., Geo, Music-20, Shopee) and 500 for large datasets (i.e., Music-200, Music-2000, Person).

Since the prediction pairs from two-table EM approaches can not be directly used for evaluation in the multi-table EM setting, we devised an extension algorithm for converting pairs into tuples with its pseudo code presented in Algorithm \ref{algo:evaluation_two_table}.

\begin{algorithm}[t]
\caption{Extension for Pairs to Tuples}
\label{algo:evaluation_two_table}
\LinesNumbered
\DontPrintSemicolon
\label{algorithmn}
    \KwIn{pairs $\mathcal{P}$, entity set $\mathbf{E}$}
    \KwOut{tuples $\mathcal{T}$ in the multi-table EM setting}
    $\mathcal{T} \gets$ empty set\;
    \For {$e \in E$}{
        \tcp{Find all entities in $\mathcal{P}$ that match $e$.}
        $E^{\prime} \gets find\_matched\_entities(\mathcal{P}, e)$ \;
        \tcp{Construct to a tuple.}
        $tuple \gets e \cup E^{\prime}$ \;
        $\mathcal{T} \gets \operatorname{Add}(\mathcal{T}, tuple)$ \;
    }
\Return{$\mathcal{T}$}
\end{algorithm}

\subsection{Experiments on Effectiveness (RQ1)}
\label{exp:effectiveness}

\begin{table*}\small
\begin{threeparttable}
\setlength{\tabcolsep}{1.5pt}
\centering
\caption{Matching performance of all the methods.}
\label{table:main_result}
\begin{tabular}
{c|cccc|cccc|cccc|cccc|cccc|cccc}
\toprule
\multirow{2}{*}{\textbf{Methods}} & \multicolumn{4}{c|}{\textbf{Geo}} & \multicolumn{4}{c|}{\textbf{Music-20}} & \multicolumn{4}{c|}{\textbf{Music-200}} & \multicolumn{4}{c|}{\textbf{Music-2000}} & \multicolumn{4}{c|}{\textbf{Person}} & \multicolumn{4}{c}{\textbf{Shopee}} \\ \cmidrule{2-25}
 & P & R & F1 & p-F1 & P & R & F1 & p-F1 & P & R & F1 & p-F1 & P & R & F1 & p-F1 & P & R & F1 & p-F1 & P & R & F1 & p-F1 \\ 
 \midrule
\textsf{PromptEM (pw)} & 13.0 & 48.2 & 20.4 & 55.2 & 31.7 & 83.0 & 45.9 & 70.7 & 21.9 & 68.6 & 33.2 & 55.3 & \textbackslash{} & \textbackslash{} & \textbackslash{} & \textbackslash{} & \textbackslash{} & \textbackslash{} & \textbackslash{} & \textbackslash{} & 0.2 & 0.5 & 0.2 & 9.6 \\
\textsf{Ditto (pw)} & 11.2 & 41.5 & 17.6 & 30.4 & 39.4 & 85.5 & 53.9 & 70.9 & 28.7 & 75.6 & 41.6 & 56.1 & \textbackslash{} & \textbackslash{} & \textbackslash{} & \textbackslash{} & \textbackslash{} & \textbackslash{} & \textbackslash{} & \textbackslash{} & 0.0 & 0.0 & 0.0 & 2.0 \\
\textsf{AutoFJ (pw)} & 95.1 & 42.4 & 58.6 & 89.4 & 70.7 & 4.5 & 8.4 & 56.6 & - & - & - & - & - & - & - & - & - & - & - & - & 71.1 & 10.8 & 18.7 & \textbf{45.0} \\  \midrule
\textsf{PromptEM (c)} & 33.7 & 88.0 & 48.7 & 85.3 & 41.1 & 92.3 & 56.9 & 78.9 & 29.4 & 80.4 & 43.0 & 64.3 & \textbackslash{} & \textbackslash{} & \textbackslash{} & \textbackslash{} & \textbackslash{} & \textbackslash{} & \textbackslash{} & \textbackslash{} & 2.2 & 6.6 & 3.3 & 22.0 \\
\textsf{Ditto (c)} & 24.0 & 76.6 & 36.5 & 65.7 & 48.4 & 91.5 & 63.3 & 76.8 & 40.9 & 87.8 & 55.8 & 72.6 & \textbackslash{} & \textbackslash{} & \textbackslash{} & \textbackslash{} & \textbackslash{} & \textbackslash{} & \textbackslash{} & \textbackslash{} & 3.4 & 10.0 & 5.1 & 19.6 \\
\textsf{AutoFJ (c)} & 52.3 & 50.0 & 51.1 & 56.8 & 30.3 & 23.4 & 26.4 & 50.4 & - & - & - & - & - & - & - & - & - & - & - & - & 45.9 & 24.2 & \textbf{31.6} & 31.1 \\  \midrule
\textsf{ALMSER-GB} & 34.0 & 85.4 & 48.6 & 83.8 & 48.6 & 91.5 & 63.5 & 87.0 & \textbackslash{} & \textbackslash{} & \textbackslash{} & \textbackslash{} & \textbackslash{} & \textbackslash{} & \textbackslash{} & \textbackslash{} & \textbackslash{} & \textbackslash{} & \textbackslash{} & \textbackslash{} & 7.9 & 22.3 & 11.7 & 36.4 \\
\textsf{MSCD-HAC} & 39.0 & 91.0 & 54.6 & 90.9 & \textbackslash{} & \textbackslash{} & \textbackslash{} & \textbackslash{} & \textbackslash{} & \textbackslash{} & \textbackslash{} & \textbackslash{} & \textbackslash{} & \textbackslash{} & \textbackslash{} & \textbackslash{} & \textbackslash{} & \textbackslash{} & \textbackslash{} & \textbackslash{} & \textbackslash{} & \textbackslash{} & \textbackslash{} & \textbackslash{} \\  \midrule
\textsf{MultiEM} & 90.5 & 91.4 & \textbf{90.9} & \textbf{97.3} & 91.1 & 86.2 & \textbf{88.6} & \textbf{95.3} & 83.7 & 80.8 & \textbf{82.2} & \textbf{92.3} & 69.4 & 68.1 & \textbf{68.7} & \textbf{85.2} & 33.6 & 39.9 & \textbf{36.5} & \textbf{73.6} & 34.5 & 21.1 & 26.2 & 43.5 \\
\textsf{w/o EER} & 65.1 & 64.3 & 64.7 & 89.5 & 88.3 & 85.3 & 86.8 & 94.2 & 79.4 & 76.6 & 78.0 & 89.9 & 65.1 & 60.6 & 62.8 & 81.3 & 33.6 & 39.9 & \textbf{36.5} & \textbf{73.6} & 34.5 & 21.1 & 26.2 & 43.5 \\
\textsf{w/o DP} & 90.5 & 91.4 & \textbf{90.9} & \textbf{97.3} & 82.0 & 82.8 & 82.4 & 92.7 & 75.5 & 77.2 & 76.4 & 89.8 & 65.6 & 66.4 & 66.0 & 84.1 & 33.6 & 39.9 & \textbf{36.5} & \textbf{73.6} & 32.9 & 21.1 & 25.7 & 42.9 \\ \bottomrule
\end{tabular}
\begin{tablenotes}
\footnotesize
\item[1] Due to the space limitation, we use ``p-F1'' to represent ``pair-F1'' described in Section \ref{exp:setup}. And we use the suffix ``(pw)'' to indicate the pairwise matching and the suffix ``(c)'' to indicate the chain matching for two-table EM methods.
\item[2] The best ``F1'' and ``pair-F1'' are in \textbf{bold}.
\item[3] The symbol ``-''  means that the method is \textbf{NOT} able to perform due to the memory limitation in our experimental settings.
\item[4] The symbol ``$\backslash$'' denotes that the method can \textbf{NOT} produce any result after \textbf{7} days in our experimental settings.
\end{tablenotes}
\end{threeparttable}
\end{table*}

We first evaluate the matching performance of \textsf{MultiEM} compared to the baselines. The results of all methods across the six datasets are reported in Table \ref{table:main_result}.

\noindent\textbf{MultiEM vs. two-table EM methods.} As observed, \textsf{MultiEM} significantly outperforms all two-table baselines on most datasets. On datasets Geo, Music-20, and Shopee, the average F1 score of \textsf{MultiEM} is +21.7 over the respective best two-table EM competitor (i.e., \textsf{AutoFJ (p), \textsf{Ditto (c)}}, and \textsf{PromptEM (c)}). \textsf{PromptEM} and \textsf{Ditto} perform relatively well on some datasets (e.g., 63.3 of F1 score on Music-20 and 85.3 of pair-F1 on Geo ) because they utilize the pre-trained language models, which capture better entity representation than other baselines. \textsf{AutoFJ} also achieves promising results on some datasets (e.g., Geo and Shopee) while poorly on some other datasets, and even cannot produce any results on large datasets due to memory constraints. However, these two-table EM methods need to be extended by pairwise matching or chain matching, which explicitly encounter transitive conflicts (described in Challenge \uppercase\expandafter{\romannumeral 2}), hindering the effectiveness of these methods. By comparison, the extensions of chain matching perform better than pairwise matching in most cases. Specifically, the F1 score of the former is +11.2 over the latter, and the pair-F1 is +7.5.  The main reason is that the chain matching may output fewer matched pairs, and thus fewer transitive conflicts. Furthermore, we observe that for \textsf{PromptEM} and \textsf{Ditto}, the recall substantially exceeds the precision on all datasets. That is because we simplify the evaluation by taking all ground truth pairs as a part of the candidate entity pairs of these two methods.


\noindent\textbf{MultiEM vs. multi-table EM methods.} Although those baselines (i.e., \textsf{ALMSER-GB} and \textsf{MSCD-HAC}) make some designs for multi-table EM and achieve relatively considerable results on some datasets (e.g., \textsf{MSCD-HAC} scores pair-F1 of 90.9 on Geo, \textsf{ALMSER-GB} scores pair-F1 of 36.4 on Shopee), they perform poorly in terms of efficiency. \textsf{MSCD-HAC} cannot produce valid results on most datasets, and \textsf{ALMSER-GB} cannot either on large-scale datasets. In addition, \textsf{ALMSER-GB} and \textsf{MSCD-HAC} regard multi-table EM as a pairwise matching task, so they could perform better on the pair-F1 score than the F1 score.

Overall, the proposed \textsf{MultiEM} outperforms baselines in matching effectiveness across six benchmark datasets. Specifically, on four comparable datasets (i.e., Geo, Music-20, Music-200, Shopee), \textsf{MultiEM} scores an average F1 of 72.0, which is +37.0 relatively over competitive baselines, and scores an average pair-F1 of 82.1, which is +25.2 over baselines. For the two large datasets Music-2000 and Person, \textsf{MultiEM} scores an average F1 of 52.6 and pair-F1 of 79.4. However, no baselines can generate valid results due to time or memory constraints. The excellent matching performance demonstrates the effectiveness of our proposed \textsf{MultiEM}.

For dataset Shopee, we observe that all baselines and our proposed \textsf{MultiEM} have low F1 and pair-F1 scores (i.e., the maximum is 31.6 and 45.0, respectively). The main reason is that this dataset includes many similar and confusing product descriptions, so it is difficult. For example, given two different products with descriptions ``Paket Senter mini XPE+COB led Q5 zoom usb charger'' and ``Senter Mini XPE+Led COB Cas USB Zoom Police U3''. Their cosine similarity is 0.77 based on Sentence-BERT and 0.71 based on Glove \cite{pennington2014glove} embeddings. In other words, most representation models confuse them without supervised guidance. More specifically, whether it is a supervised or unsupervised method, whether it is a two-table or a multi-table EM method, one of the most critical steps is representing the entities. High-quality representations will affect the effectiveness of the downstream task~\cite{deng2022turl,yin2020tabert,yamada2020luke}. Currently, the approaches for entity representation are mainly based on word embedding \cite{mudgal2018deep}, pre-trained language models \cite{li2020deep}, or integrated with additional information \cite{ge2021collaborem} (e.g., graph, external knowledge). These methods are still flawed and perform poorly in the face of indistinguishable entity text.

\subsection{Experiments on Efficiency (RQ2)}
\label{exp:efficiency}

\begin{table}\small
\caption{Running time comparison.}
\label{table:time}
\setlength{\tabcolsep}{2pt}
\begin{threeparttable}
\begin{tabular}{c|c|c|c|c|c|c}
\toprule
\textbf{Methods} & \textbf{Geo} & \begin{tabular}[c]{@{}c@{}}\textbf{Music-}\\\textbf{20}\end{tabular} & \begin{tabular}[c]{@{}c@{}}\textbf{Music-}\\\textbf{200}\end{tabular} & \begin{tabular}[c]{@{}c@{}}\textbf{Music-}\\\textbf{2000}\end{tabular} & \textbf{Person} & \textbf{Shopee} \\ 
\midrule
\textsf{PromptEM (pw)} & 12.7m & 50.5m & 38.4h & $\backslash$ & $\backslash$ & 3.0h \\
\textsf{Ditto (pw)} & 3.5m & 31.4m & 14.4h & $\backslash$ & $\backslash$ & 1.6h \\
\textsf{AutoFJ (pw)} & 8.9m & 3.8h & - & - & - & 3.1h \\ 
\midrule
\textsf{PromptEM (c)} & 12.1m & 49.8m & 39.4h & $\backslash$ & $\backslash$ & 2.6h \\
\textsf{Ditto (c)} & 3.4m & 31.2m & 14.5h & $\backslash$ & $\backslash$ & 1.5h \\
\textsf{AutoFJ (c)} & 9.9m & 1.4h & - & - & - & 1.2h \\ 
\midrule
\textsf{ALMSER-GB} & 5.1m & 21.0m & $\backslash$ & $\backslash$ & $\backslash$ & 26.8m \\
\textsf{MSCD-HAC} & 1.5h & $\backslash$ & $\backslash$ & $\backslash$ & $\backslash$ & $\backslash$ \\
\midrule
\textsf{MultiEM} & \textbf{6.1s} & 34.6s & 6.3m & 1.3h & 1.8h & 42.9s \\
\textsf{MultiEM (parallel)} & 10.7s & \textbf{31.0s} & \textbf{4.2m} & \textbf{49.1m} & \textbf{52.9m} & \textbf{31.8s} \\ 
\bottomrule
\end{tabular}
\begin{tablenotes}
\footnotesize
\item[1] ``s'' denotes seconds, ``m'' means minutes, ``h'' denotes hours.
\item[2] The minimum running time is in \textbf{bold}.
\item[3] Due to the space limitation, we use the suffix ``(pw)'' to indicate the pairwise matching and the suffix ``(c)'' to indicate the chain matching for two-table EM methods.
\item[4] The symbol ``-''  means that the method is \textbf{NOT} able to perform due to the memory limitation in our experimental settings.
\item[5] The symbol ``$\backslash$'' denotes that the method can \textbf{NOT} produce any result after \textbf{7} days in our experimental settings.
\end{tablenotes}
\end{threeparttable}
\end{table}

We further explore the efficiency of our proposed \textsf{MultiEM} in terms of running time and memory usage, and the results are presented in Table \ref{table:time}, Table \ref{table:memory}.

\noindent\textbf{Comparison of running time.} As observed, \textsf{MultiEM} and its parallelized variant \textsf{MultiEM} (parallel) show substantial advantages in terms of running time. \textsf{MultiEM} achieves state-of-the-art matching results with nearly 170x speed-up on average compared to competitors and over 190x speed-up for \textsf{MultiEM} (parallel). On datasets Geo Music-20, and Shopee, the running time of \textsf{MultiEM} is at the second level, while other baselines are at the minute or even hour level. On large-scale datasets such as Music-2000 and Person, most baselines cannot produce any results due to the time limitation, which highlights the high efficiency of \textsf{MultiEM}. Generally, two-table EM methods (i.e., \textsf{PromptEM}, \textsf{Ditto}, and \textsf{AutoFJ}) run long as they are not explicitly designed for multi-table EM, requiring pairwise or chain matching extensions. Among them, \textsf{Ditto} runs long because it needs to fine-tune the pre-trained language model. And \textsf{PromptEM} also takes longer to run as it needs to handle the prompt-tuning template, which is more complex than vanilla fine-tuning (i.e., \textsf{Ditto}). In addition, it is observed that the running time of chain matching is near to pairwise matching, thereby also inefficient. As said before, this is because the size of the base table increases with the chain matching, which significantly affects the matching efficiency. \textsf{MSCD-HAC} is based on agglomerative hierarchical clustering, and its time complexity is too high, i.e., $O(\lvert E\rvert ^3)$, where $E$ represents all entities. Therefore, \textsf{MSCD-HAC} cannot support large-scale datasets.
\textsf{ALMSER-GB} applies active learning and boosted graph learning, which is ahead of other baselines in the running time, but still cannot handle some large datasets.

\noindent\textbf{Comparison of memory usage.} In terms of memory usage, \textsf{MultiEM} is relatively low on most datasets, including some large-scale datasets. 
The reason is that \textsf{MultiEM} is based on the approximate k-nearest neighbor (ANN) and does not depend on any large or complex models, which usually occur in lots of memory. For methods such as \textsf{PromptEM} and \textsf{Ditto} that rely on pre-trained language models, their memory usage is the highest and generally stable regardless of dataset size. \textsf{AutoFJ} also has low memory usage on small datasets. However, the blocking phase on large datasets causes a surge in memory usage, so it cannot produce valid results due to memory limitations. \textsf{ALMSER-GB} needs to store and process the entity similarity graphs, so the memory usage of it varies due to the number of entities.

\begin{table}\small
\caption{Memory usage comparison.}
\label{table:memory}
\setlength{\tabcolsep}{2pt}
\begin{threeparttable}
\begin{tabular}{c|c|c|c|c|c|c}
\toprule
\textbf{Methods} & \textbf{Geo} & \begin{tabular}[c]{@{}c@{}}\textbf{Music-}\\\textbf{20}\end{tabular} & \begin{tabular}[c]{@{}c@{}}\textbf{Music-}\\\textbf{200}\end{tabular} & \begin{tabular}[c]{@{}c@{}}\textbf{Music-}\\\textbf{2000}\end{tabular} & \textbf{Person} & \textbf{Shopee} \\ 
\midrule
\textsf{PromptEM (pw)} & 43.9G & 43.9G & 65.5G & $\backslash$ & $\backslash$ & 39.2G \\
\textsf{Ditto (pw)} & 30.1G & 41.6G & 44.1G & $\backslash$ & $\backslash$ & 68.6G \\
\textsf{AutoFJ (pw)} & 5.1G & \textbf{6.7G} & - & - & - & \textbf{3.0G}
\\
\midrule
\textsf{PromptEM (c)} & 43.4G & 44.4G & 65.5G & $\backslash$ & $\backslash$ & 39.5G \\
\textsf{Ditto (c)} & 30.4G & 40.7G & 44.3G & $\backslash$ & $\backslash$ & 68.5G \\
\textsf{AutoFJ (c)} & 5.3G & 7.0G & - & - & - & \textbf{3.0G}
\\
\midrule
\textsf{ALMSER-GB} & 3.8G & 15.7G & $\backslash$ & $\backslash$ & $\backslash$ & 9.9G \\
\textsf{MSCD-HAC} & \textbf{2.1G} & $\backslash$ & $\backslash$ & $\backslash$ & $\backslash$ & $\backslash$ \\
\midrule
\textsf{MultiEM} & 16.3G & 17.5G & \textbf{17.8G} & \textbf{17.5G} & \textbf{18.2G} & 17.5G \\
\textsf{MultiEM (parallel)} & 21.5G & 22.1G & 23.3G & 22.0G & 24.7G & 22.7G \\
\bottomrule
\end{tabular}
\begin{tablenotes}
\footnotesize
\item[1] ``G'' denotes gigabytes.
\item[2] The minimum memory usage is in \textbf{bold}. 
\item[3] Due to the space limitation, we use the suffix ``(pw)'' to indicate the pairwise matching and the suffix ``(c)'' to indicate the chain matching for two-table EM methods.
\item[4] For \textsf{PromptEM} and \textsf{Ditto}, we report the sum of memory and GPU memory.
\item[5] The symbol ``-''  means that the method is \textbf{NOT} able to perform due to the memory limitation in our experimental settings.
\item[6] The symbol ``$\backslash$'' denotes that the method can \textbf{NOT} produce any result after \textbf{7} days in our experimental settings.
\end{tablenotes}
\end{threeparttable}
\end{table}

\subsection{Ablation Study (RQ3)}
\label{exp:ablation}

Next, we study the effectiveness and efficiency of each key module of \textsf{MultiEM}. Specifically, we analyze the effectiveness of the \emph{enhanced entity representation (EER)} and \emph{density-based pruning (DP)} modules by comparing \textsf{MultiEM} with its variants (i.e., \textsf{MultiEM w/o EER} and \textsf{MultiEM w/o DP}). The results are listed in Table \ref{table:main_result}. Furthermore, we also analyze the impact of parallel extension on overall efficiency by comparing running time and memory usage. The results are listed in Table \ref{table:time} and Table \ref{table:memory}. Finally, we evaluate the contribution of each module of \textsf{MultiEM} in terms of the running time. The results are shown in Figure \ref{fig:time_of_module}.

\begin{table}
\renewcommand\arraystretch{1.25}
\setlength{\tabcolsep}{10pt}
\caption{automated selected attributes.}
\label{table:selections}
\begin{tabular}{l|l|l}
\toprule
\textbf{Dataset} & \textbf{All attributes} & \textbf{Selected attributes} \\ \hline
\textbf{Geo} & name, longtitude, latitude & name \\ \hline
\textbf{Music-20} & \begin{tabular}[c]{@{}l@{}}id, number, title, length, \\ artist, album, year, language\end{tabular} & title, artist, album \\ \hline
\textbf{Music-200} & \begin{tabular}[c]{@{}l@{}}id, number, title, length, \\ artist, album, year, language\end{tabular} & title, artist, album \\ \hline
\textbf{Music-2000} & \begin{tabular}[c]{@{}l@{}}id, number, title, length, \\ artist, album, year, language\end{tabular} & title, artist, album \\ \hline
\textbf{Person} & \begin{tabular}[c]{@{}l@{}} givenname,surname, \\ suburb,postcode\end{tabular} & \begin{tabular}[c]{@{}l@{}} givenname,surname, \\ suburb,postcode\end{tabular} \\ \hline
\textbf{Shopee} & title & title \\
\bottomrule
\end{tabular}
\end{table}

\noindent\textbf{\textsf{MultiEM} vs. \textsf{MultiEM w/o EER}.} \textsf{MultiEM w/o EER} means that we only use the pre-trained Sentence-BERT embeddings as the final representation of entities. As demonstrated by the experimental results, the absence of the enhanced entity representation significantly decreases the matching performance, resulting in an average F1 score decrease of 6.4\% and an average pair-F1 decrease of 2.5\%. These findings suggest that the proposed enhanced entity representation module improves the entity representation quality and thus boosts the matching performance. Moreover, these results also indicate the importance of entity representation in the EM task. In addition, the selected attributes by the EER module, which are consistent with the judgments obtained by domain experts after analyzing the data, are shown in Table \ref{table:selections}.

\noindent\textbf{\textsf{MultiEM} vs. \textsf{MultiEM w/o DP}.} \textsf{MultiEM w/o DP} denotes that we only use the predictions of the merging phase as the final results. It is observed that the pruning phase contributes to performance gain in most cases. 
The F1 score drops by 2.4\%, and the pair-F1 drops by 1.1\% on average without the pruning module. 
This confirms that the proposed density-based pruning module can help further refine the predictions of the merging phase to produce more precise matching results.

\noindent\textbf{\textsf{MultiEM} vs. \textsf{MultiEM} (parallel).} We extended \textsf{MultiEM} with parallelization to further improve its efficiency. Our observations show that the parallel strategy significantly reduces the running time without compromising the matching performance. This is attributed to the design of \textsf{MultiEM}, where the merging of each table pair and the pruning of each tuple are independent processes. Moreover, memory usage also increases as parallel processes require additional resources for maintenance. As shown in Table \ref{table:time}, and Table \ref{table:memory}, the average running time is reduced by 32.2\%, and the average memory usage is increased by 29.7\% for all datasets except Geo. As described above, the dataset Geo's size is relatively small, so it is fast enough for the non-parallel \textsf{MultiEM}, while the parallel strategy will bring additional overhead.

\noindent\textbf{Efficiency of each module.} As shown in Figure \ref{fig:time_of_module}, merging is the most time-consuming step in most cases, which takes about 37.3\% on average of the overall pipeline, while 29.0\%, 13.5\%, and 20.2\% for the other three modules, respectively. In addition, the parallel strategy significantly improves the efficiency of the merging and the pruning phase. The running time drops by 13.8\% and 50.0\% on average of all datasets except Geo.

\begin{figure}
    \centering
    \subfigure[Geo]{
    \includegraphics[width=1.55in]{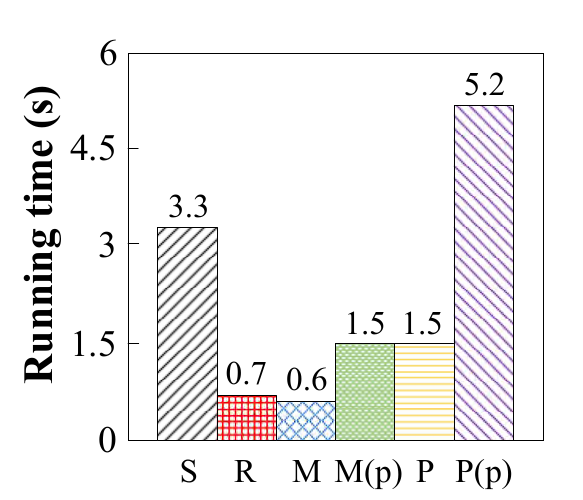}
    }
    \subfigure[Music-20]{
    \includegraphics[width=1.55in]{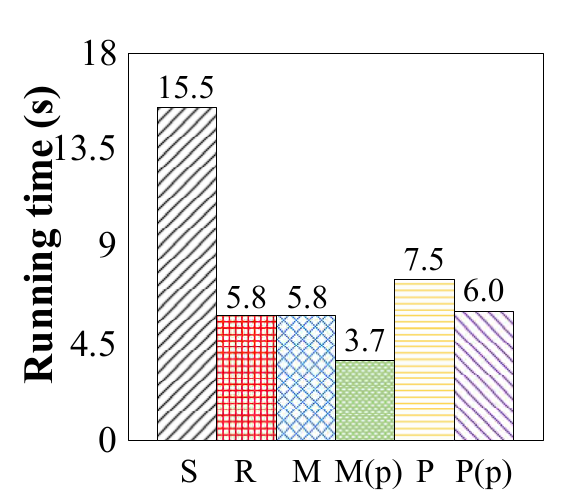}
    }\\
    \subfigure[Music-200]{
    \includegraphics[width=1.55in]{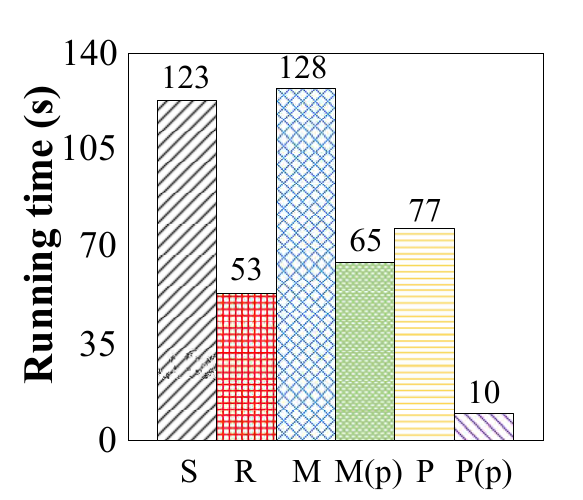}
    }
    \subfigure[Music-2000]{
    \includegraphics[width=1.55in]{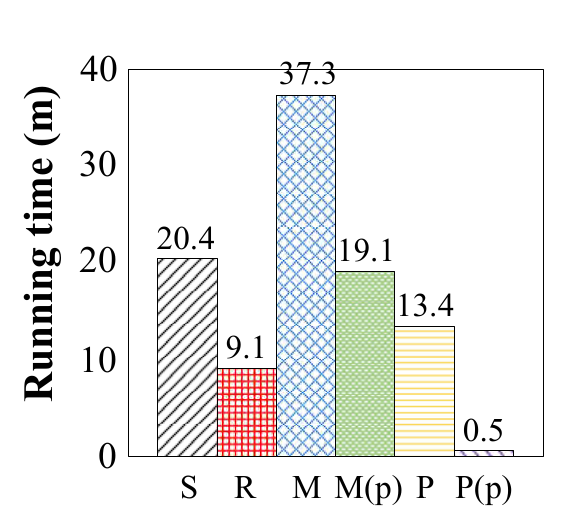}
    }\\
    \subfigure[Person]{
    \includegraphics[width=1.55in]{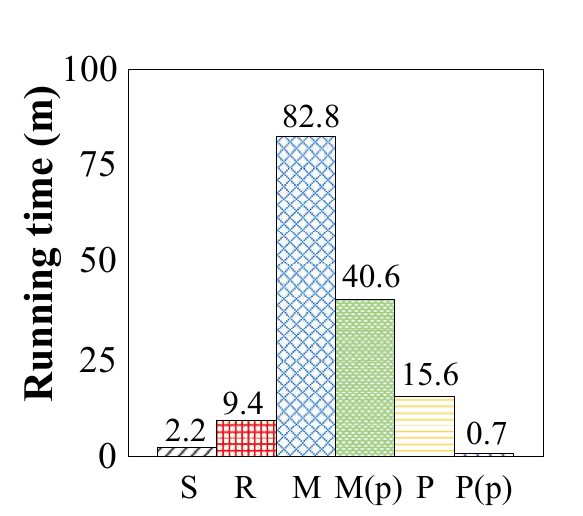}
    }
    \subfigure[Shopee]{
    \includegraphics[width=1.55in]{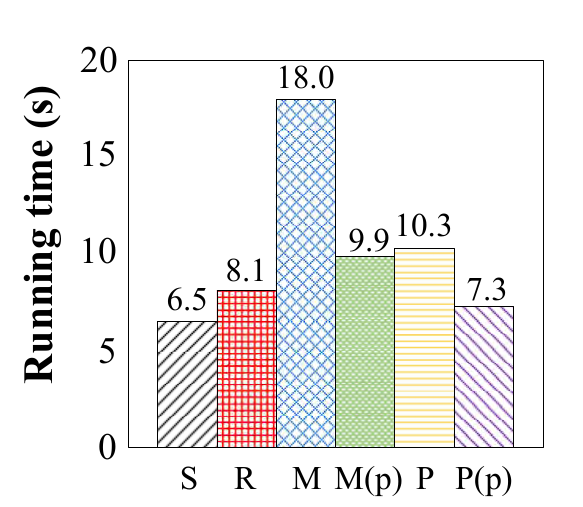}
    }\\
    \caption{Running time of each  key module of \textsf{MultiEM}. Due to the space limitation, we use abbreviations. ``S'' represents \emph{automated attribute selection},  ``R'' denotes \emph{entity representation}, ``M'' represents \emph{merging}, ``P'' denotes \emph{pruning}, and ``(p)'' represents \emph{merging/pruning in parallel}.}
    \label{fig:time_of_module}
\end{figure}

\subsection{Sensitivity (RQ4)}

We further study the sensitivity of the primary hyperparameters of the proposed \textsf{MultiEM} through the following experiments. Since the range of the running time of different datasets is too wide, following \cite{min2020emogi,chen2020pangolin}, we normalize the running time to show its variation trend better.

\noindent\textbf{Influence of $\gamma$.} We conduct a sensitivity analysis on the threshold $\gamma$ described in Section \ref{method:representation}. The results are shown in Figure \ref{fig:exp:sensitivity}(a). It is observed that as $\gamma$ varies, there are corresponding changes in the matching performance of \textsf{MultiEM}. This is because the value of $\gamma$ affects the selection of significant attributes and thus the entity representation, which is a key factor for the effectiveness of unsupervised entity matching.

\noindent\textbf{Sensitivity to the merging order.} We select four different random seeds \{0, 1, 2, 3\} and repeated the experiments on all datasets. The results are shown in Figure \ref{fig:exp:sensitivity}(b). As observed, our proposed method is not sensitive to the order of tables in the merging phase. The average variation in F1 scores is only 1.4 across all datasets. This finding can be attributed to the fact that in hierarchical merging, every entity will likely compare with another entity at some hierarchy. As a result, the order has little effect on the overall results.

\noindent\textbf{Sensitivity to $m$.} We conduct a sensitivity analysis of the distance threshold $m$ described in Section \ref{method:merging}. It is observed that the matching performance of \textsf{MultiEM} is sensitive to $m$ since the table-wise hierarchical merging strategy of \textsf{MultiEM} relies on the similarity of the entities. Therefore, we choose the optimal $m$ within the range described in Section \ref{exp:setup} for each dataset. In addition, the running time decreases slightly with the increase of $m$ due to the reduced merged pairs.

\noindent\textbf{Sensitivity to $\epsilon$.} We perform a sensitivity analysis of the clustering radius $\epsilon$ in Eq. \ref{equ:dbscan_eps_core} and Eq \ref{equ:dbscan_eps_reachable}. The results are reported in Figures \ref{fig:exp:sensitivity}(e) and \ref{fig:exp:sensitivity}(f). We find that the overall matching performance is stable as the $\epsilon$ varies. In some cases, the F1 score increases when $\epsilon$ increases; in others, it drops. That is because a smaller $\epsilon$ leads to more false outlier entities, while a larger $\epsilon$ will cause misjudgment of some core entities and reachable entities. In addition, we find that the running time of \textsf{MultiEM} under different $\epsilon$ is stable since $\epsilon$ only affects the correctness of pruning and not the number of pruning operations.

\begin{figure}
    \centering
\includegraphics[width=0.5\textwidth]{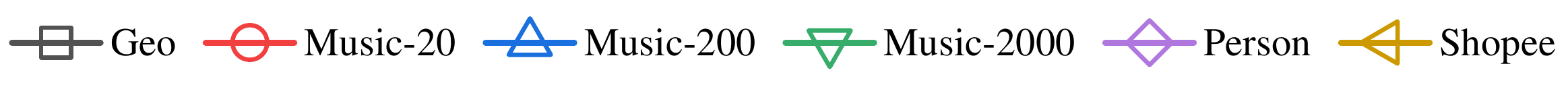}\\
    \subfigure[\textsf{MultiEM}]{
    \includegraphics[width=1.55in]{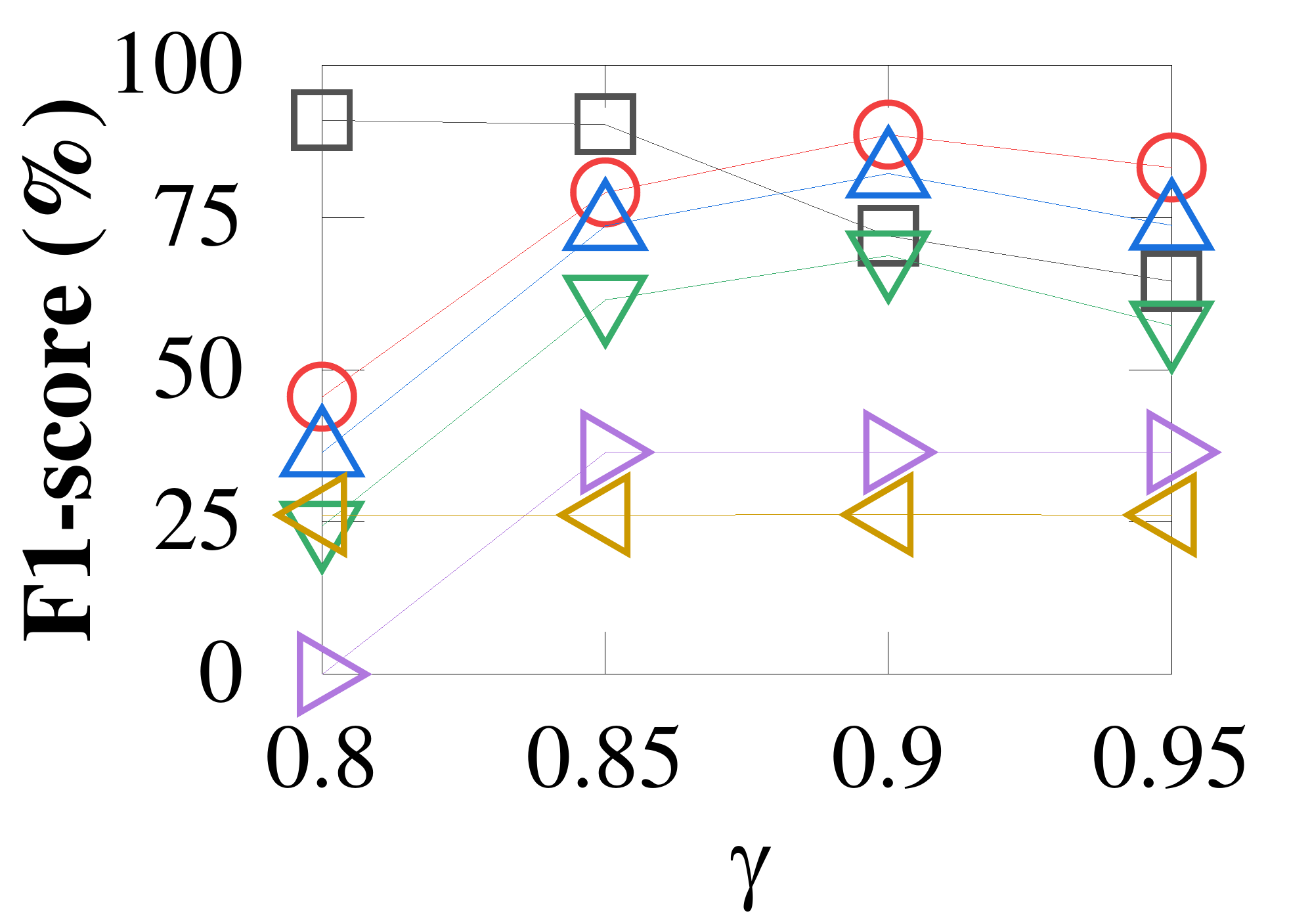}
    }
    \subfigure[\textsf{MultiEM}]{
    \includegraphics[width=1.55in]{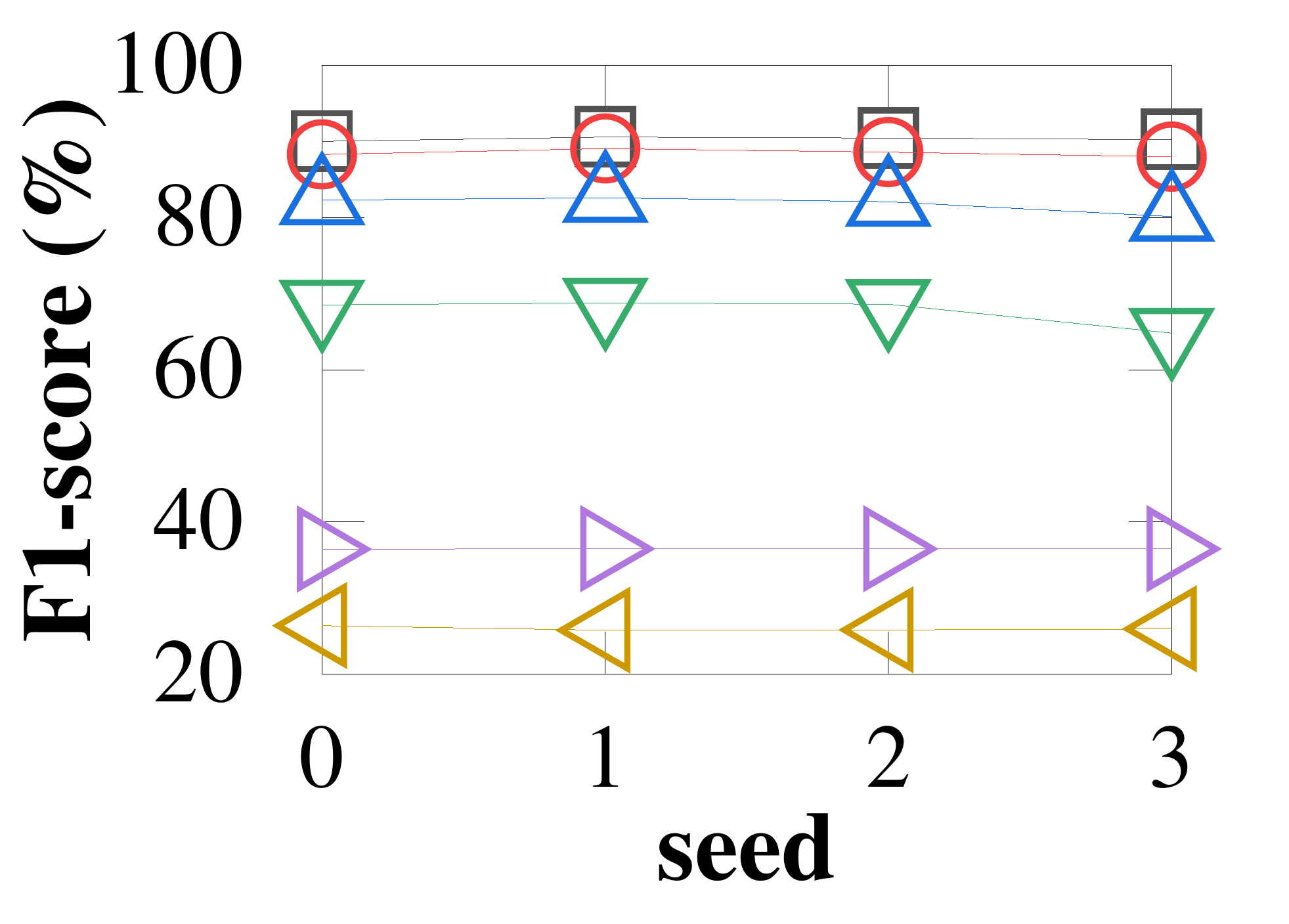}
    }\\
    \subfigure[\textsf{MultiEM}]{
    \includegraphics[width=1.55in]{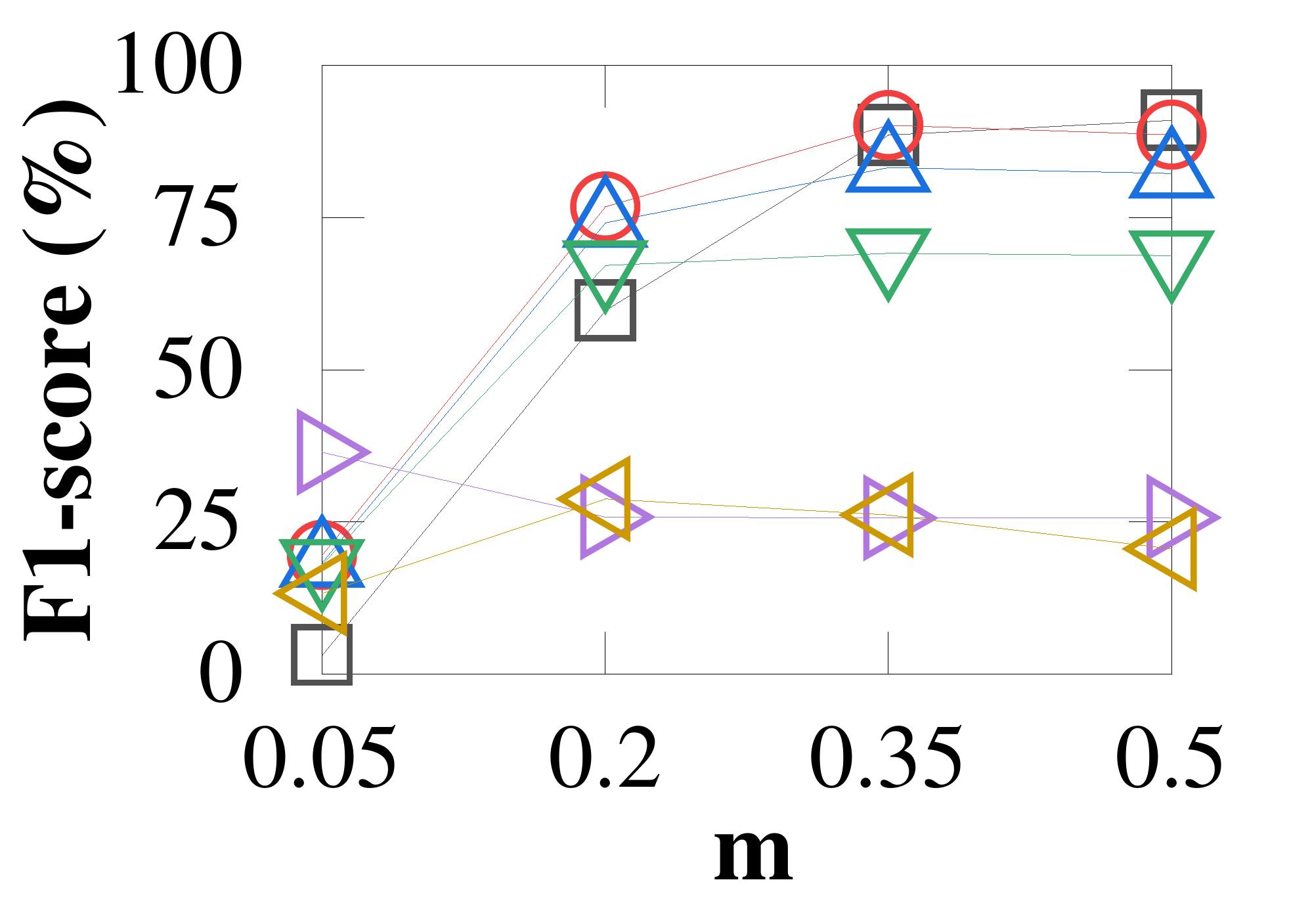}
    }
    \subfigure[\textsf{MultiEM}]{
    \includegraphics[width=1.55in]{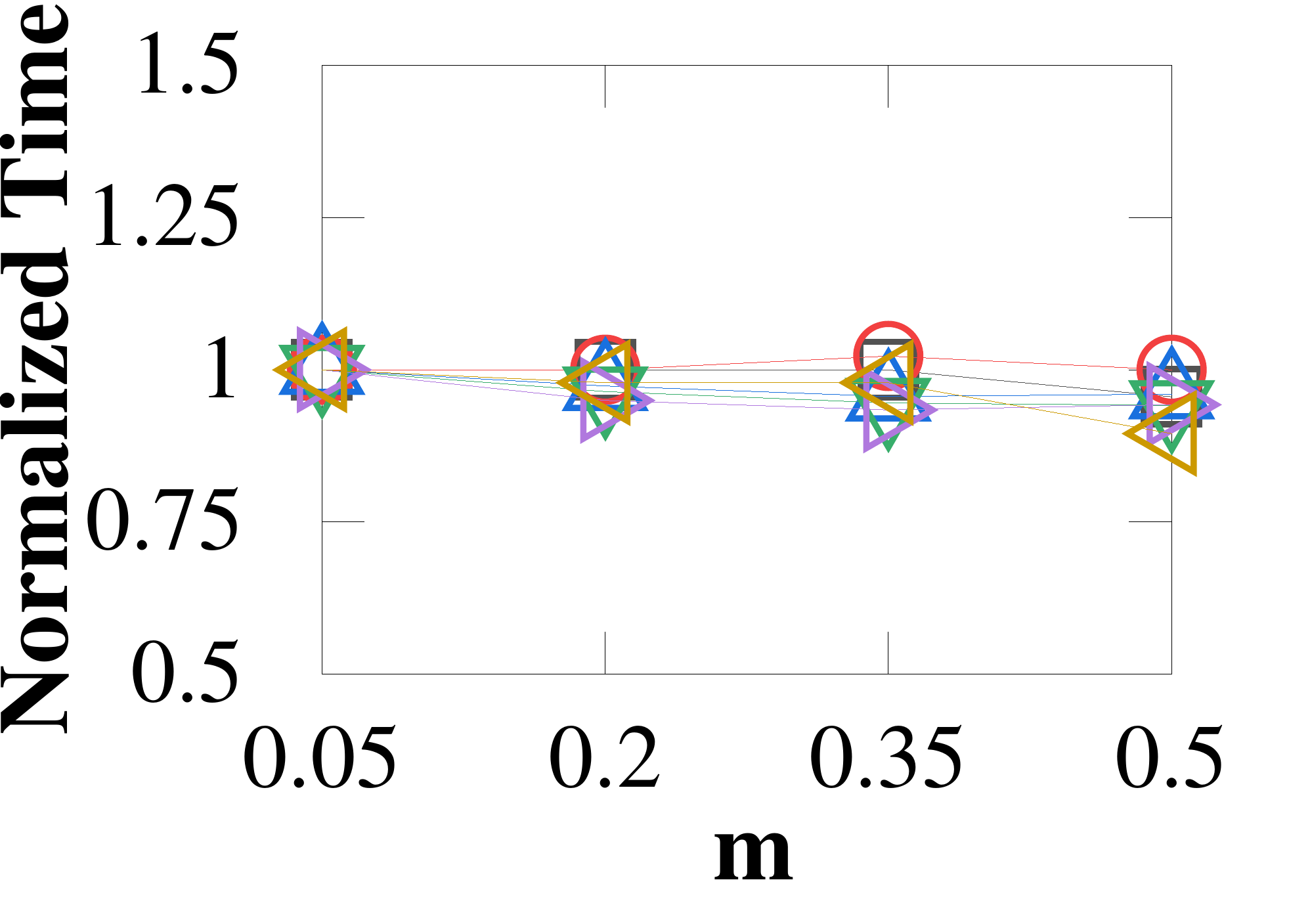}
    }\\
    \subfigure[\textsf{MultiEM}]{
    \includegraphics[width=1.55in]{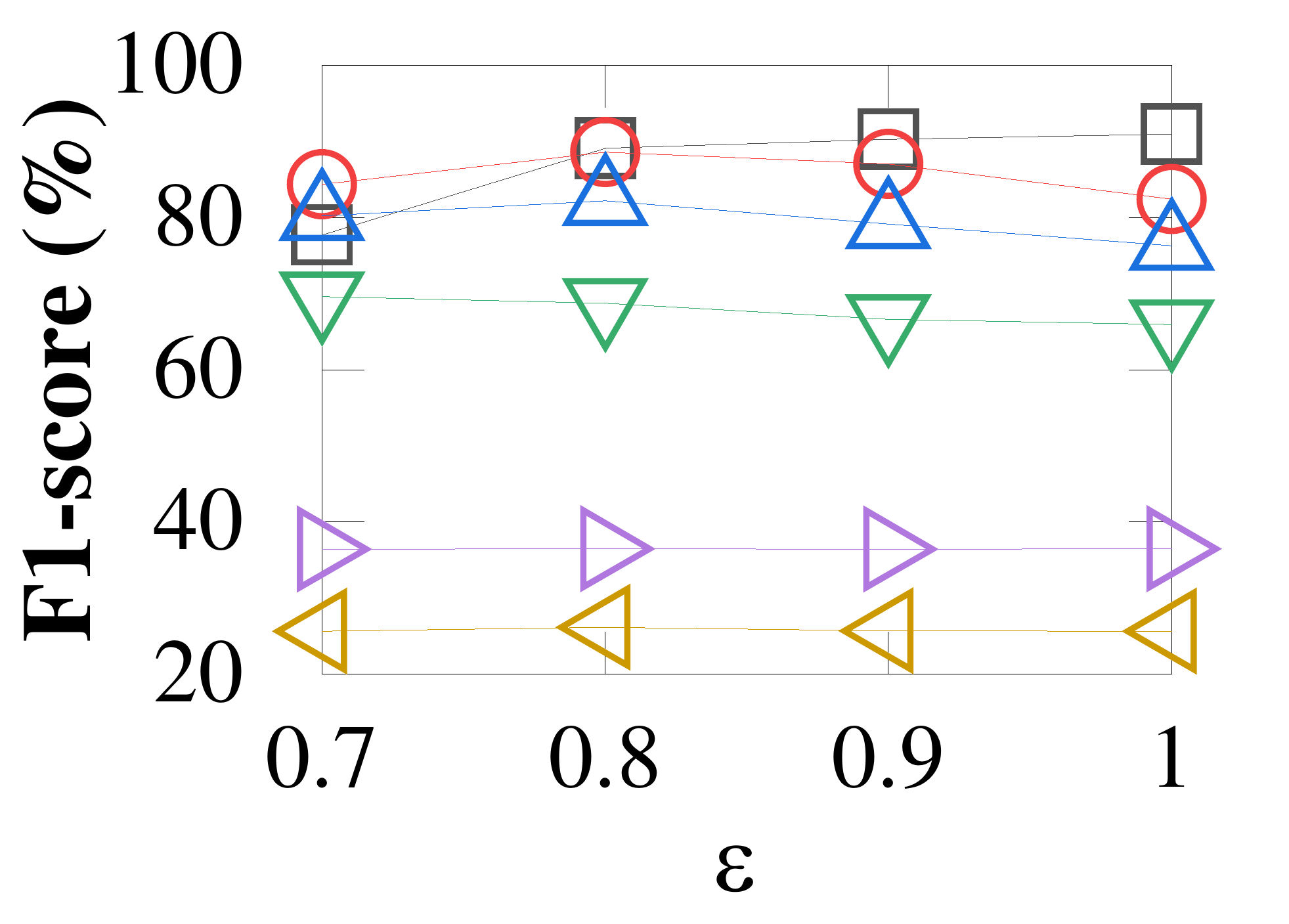}
    }
    \subfigure[\textsf{MultiEM}]{
    \includegraphics[width=1.55in]{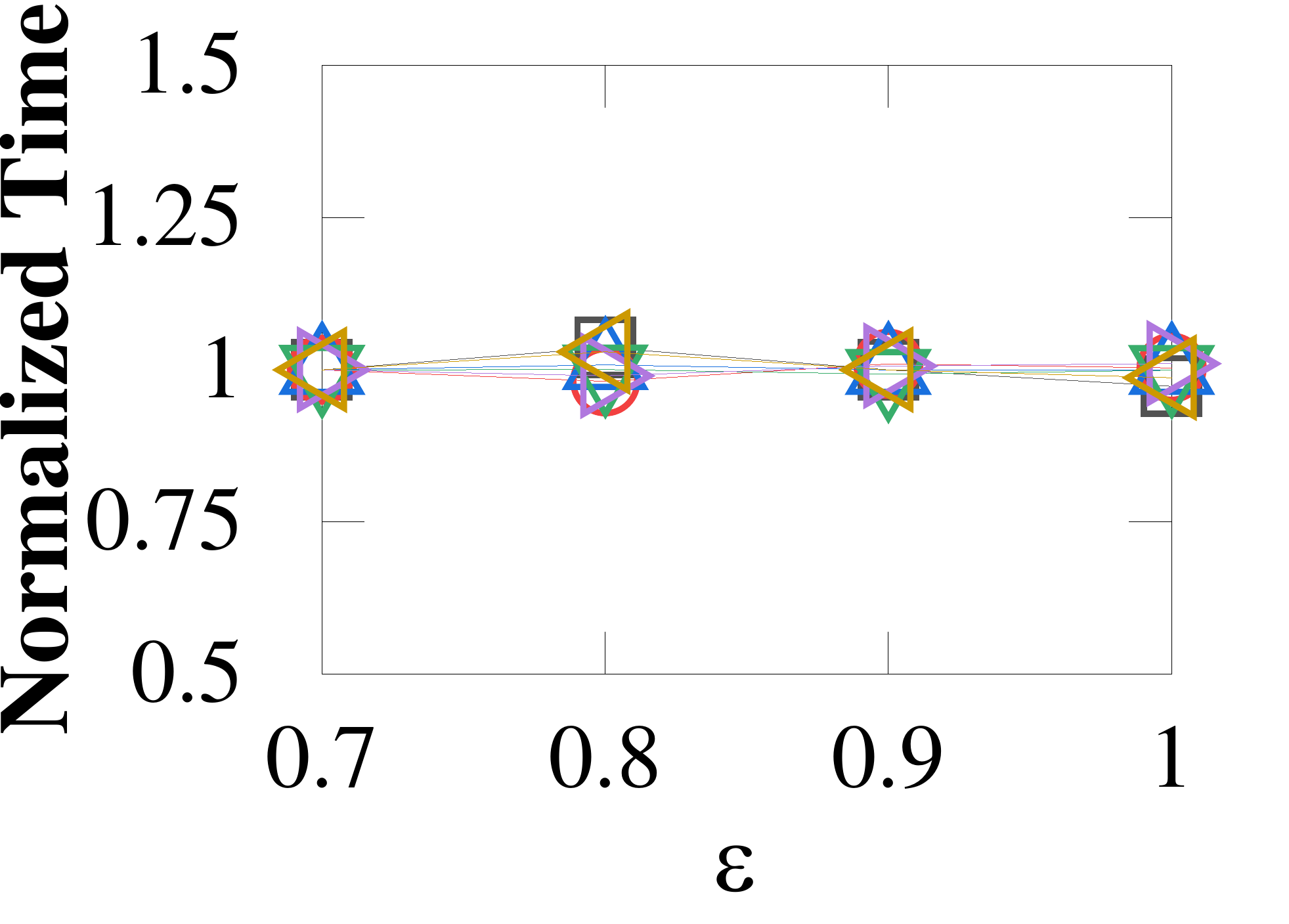}
    }\\
    \caption{Sensitivity analysis.}
    \label{fig:exp:sensitivity}
\end{figure}
\vspace{-2mm}
\section{Related work}
\label{related_work}

Entity Matching (EM) is one of the most fundamental tasks in data management, which is significant for many downstream tasks. Many practical approaches have been developed to solve this problem,  including rule-based methods \cite{elmagarmid2014nadeef,singh2017synthesizing}, crowdsourcing-based methods \cite{gokhale2014corleone, wang2012crowder}, and traditional ML-based methods~\cite{konda2016magellan,li2021auto}. 

In recent years, Deep Learning has been widely used for Entity Matching. \textsf{DeepER} \cite{ebraheem2018distributed} uses deep neural networks as feature extractors and considers EM as a binary classification task. \textsf{DeepMatcher} ~\cite{mudgal2018deep} systematically describes a space of DL solutions for EM. \textsf{Auto-EM} \cite{zhao2019auto} improves performance by pre-training the EM model with entity-type detection as an auxiliary task. \textsf{Ditto} \cite{li2020deep} first applies the pre-trained language models to EM, which gains the SOTA performance. \textsf{JointBERT} \cite{peeters2021dual} and \textsf{Sudowoodo} \cite{wang2022sudowoodo} integrate other purposes/tasks to enhance the matching performance. \textsf{FlexER} \cite{genossar2022flexer} employs  contemporary methods for universal entity resolution tasks. However, DL-based methods rely on lots of labeled samples for better performance. To this end, \textsf{Rotom} \cite{miao2021rotom} leverage meta-learning and data enhancement techniques. \textsf{CollaborEM} \cite{ge2021collaborem} designs a self-supervised framework for EM. In addition, some other studies also try to enhance the performance via active learning \cite{kasai2019low,nafa2022active}, transfer learning \cite{tu2022domain,jin2021deep,kirielle2022transer}, and other promising technologies \cite{gao2022clusterea,liu2023unsupervised,yao2022entity}.

Most EM methods are only designed for two tables, which limits their application in multi-source scenarios. Some studies \cite{lerm2021extended,nentwig2016holistic} apply clustering algorithms to multi-source entity matching. \textsf{MSCD-HAC} \cite{saeedi2021matching} proposes extensions to hierarchical agglomerative clustering to match and cluster entities from multiple sources. \textsf{MSCD-AP} \cite{lerm2021extended} regard multi-table entity matching as an affinity propagation clustering task.  \textsf{ALMSER} ~\cite{primpeli2021graph} proposes a graph-boosted active learning method for multi-source entity resolution. However, as evaluated in Sections \ref{exp:effectiveness} and \ref{exp:efficiency}, they are not effective and efficient enough. 
\vspace{1mm}
\section{Conclusions}
\label{conclusions}

For the first time, we formally study the problem of unsupervised multi-table EM and formulate it as a two-step process (i.e., \emph{merging} and \emph{pruning}). We propose an efficient and effective solution, dubbed \textsf{MultiEM}. First, we present a parallelizable table-wise hierarchical merging algorithm to match multiple tables efficiently. Furthermore, in terms of effectiveness, we enhance the entity representation quality by a novel automated attribute selection strategy and explicitly avoid the transitive conflicts by hierarchical merging. Finally, we develop a density-based strategy to prune outliers and further improve effectiveness. Extensive experimental results on six real-world datasets with various numbers of sources (i.e., from 4 to 20) demonstrate the superiority of \textsf{MultiEM} both in the effectiveness and efficiency compared with the state-of-the-art approaches. 

Based on our analysis, the main limitations of our work are twofold: (i) To ensure efficiency, we focus solely on representation-based entity matching and do not explore more effective interaction-based techniques, which are used for most SOTA supervised EM methods like \textsf{Ditto} and \textsf{PromptEM}; (ii) we overlook the merging paths of each data item in the hierarchical merging, which could be helpful for subsequent pruning.


In the future, we plan to explore a more efficient merging strategy to support larger-scale data, e.g., merging in a distributed manner. And we plan to investigate some interactive technologies, such as self-supervised learning, to enhance effectiveness. These efforts will contribute to advancing the state-of-the-art in entity matching and enable the processing of larger and more complex data in real-world applications.

\newpage
\balance
\clearpage

\bibliographystyle{IEEEtran}
\bibliography{IEEEabrv, refer}

\begin{thebibliography}{10}
\providecommand{\url}[1]{#1}
\csname url@samestyle\endcsname
\providecommand{\newblock}{\relax}
\providecommand{\bibinfo}[2]{#2}
\providecommand{\BIBentrySTDinterwordspacing}{\spaceskip=0pt\relax}
\providecommand{\BIBentryALTinterwordstretchfactor}{4}
\providecommand{\BIBentryALTinterwordspacing}{\spaceskip=\fontdimen2\font plus
\BIBentryALTinterwordstretchfactor\fontdimen3\font minus
  \fontdimen4\font\relax}
\providecommand{\BIBforeignlanguage}[2]{{%
\expandafter\ifx\csname l@#1\endcsname\relax
\typeout{** WARNING: IEEEtran.bst: No hyphenation pattern has been}%
\typeout{** loaded for the language `#1'. Using the pattern for}%
\typeout{** the default language instead.}%
\else
\language=\csname l@#1\endcsname
\fi
#2}}
\providecommand{\BIBdecl}{\relax}
\BIBdecl

\bibitem{li2020deep}
Y.~Li, J.~Li, Y.~Suhara, A.~Doan, and W.-C. Tan, ``Deep entity matching with
  pre-trained language models,'' \emph{PVLDB}, vol.~14, no.~1, pp. 50--60,
  2020.

\bibitem{wang2022promptem}
P.~Wang, X.~Zeng, L.~Chen, F.~Ye, Y.~Mao, J.~Zhu, and Y.~Gao, ``Promptem:
  prompt-tuning for low-resource generalized entity matching,'' \emph{PVLDB},
  vol.~16, no.~2, pp. 369--378, 2022.

\bibitem{li2021auto}
P.~Li, X.~Cheng, X.~Chu, Y.~He, and S.~Chaudhuri, ``Auto-fuzzyjoin:
  Auto-program fuzzy similarity joins without labeled examples,'' in
  \emph{SIGMOD}, 2021, pp. 1064--1076.

\bibitem{tu2022domain}
J.~Tu, J.~Fan, N.~Tang, P.~Wang, C.~Chai, G.~Li, R.~Fan, and X.~Du, ``Domain
  adaptation for deep entity resolution,'' in \emph{SIGMOD}, 2022, pp.
  443--457.

\bibitem{pricerunner}
\BIBentryALTinterwordspacing
PriceRunner, 2023. [Online]. Available: \url{https://www.pricerunner.com/}
\BIBentrySTDinterwordspacing

\bibitem{skroutz}
\BIBentryALTinterwordspacing
Skroutz, 2023. [Online]. Available: \url{https://www.skroutz.gr/}
\BIBentrySTDinterwordspacing

\bibitem{mudgal2018deep}
S.~Mudgal, H.~Li, T.~Rekatsinas, A.~Doan, Y.~Park, G.~Krishnan, R.~Deep,
  E.~Arcaute, and V.~Raghavendra, ``Deep learning for entity matching: A design
  space exploration,'' in \emph{SIGMOD}, 2018, pp. 19--34.

\bibitem{primpeli2021graph}
A.~Primpeli and C.~Bizer, ``Graph-boosted active learning for multi-source
  entity resolution,'' in \emph{ISWC}.\hskip 1em plus 0.5em minus 0.4em\relax
  Springer, 2021, pp. 182--199.

\bibitem{wu2020zeroer}
R.~Wu, S.~Chaba, S.~Sawlani, X.~Chu, and S.~Thirumuruganathan, ``Zeroer: Entity
  resolution using zero labeled examples,'' in \emph{SIGMOD}, 2020, pp.
  1149--1164.

\bibitem{saeedi2021matching}
A.~Saeedi, L.~David, and E.~Rahm, ``Matching entities from multiple sources
  with hierarchical agglomerative clustering.'' in \emph{KEOD}, 2021, pp.
  40--50.

\bibitem{lerm2021extended}
S.~Lerm, A.~Saeedi, and E.~Rahm, ``Extended affinity propagation clustering for
  multi-source entity resolution,'' \emph{BTW 2021}, 2021.

\bibitem{Gazzarri2022ProgressiveER}
L.~Gazzarri and M.~Herschel, ``Progressive entity resolution over incremental
  data,'' in \emph{EDBT}, 2022.

\bibitem{Papadakis2021TheFG}
G.~Papadakis, E.~Ioannou, E.~Thanos, and T.~Palpanas, ``The four generations of
  entity resolution,'' in \emph{Synthesis Lectures on Data Management}, 2021.

\bibitem{cappuzzo2020creating}
R.~Cappuzzo, P.~Papotti, and S.~Thirumuruganathan, ``Creating embeddings of
  heterogeneous relational datasets for data integration tasks,'' in
  \emph{SIGMOD}, 2020, pp. 1335--1349.

\bibitem{Yin2020TaBERTPF}
P.~Yin, G.~Neubig, W.~tau Yih, and S.~Riedel, ``Tabert: Pretraining for joint
  understanding of textual and tabular data,'' \emph{ArXiv}, vol.
  abs/2005.08314, 2020.

\bibitem{reimers2019sentence}
N.~Reimers and I.~Gurevych, ``Sentence-bert: Sentence embeddings using siamese
  bert-networks,'' in \emph{EMNLP}, 2019.

\bibitem{ge2021collaborem}
C.~Ge, P.~Wang, L.~Chen, X.~Liu, B.~Zheng, and Y.~Gao, ``Collaborem: A
  self-supervised entity matching framework using multi-features
  collaboration,'' \emph{TKDE}, 2021.

\bibitem{Wang2021MachampAG}
J.~Wang, Y.~Li, and W.~Hirota, ``Machamp: A generalized entity matching
  benchmark,'' \emph{CIKM}, 2021.

\bibitem{li2019approximate}
W.~Li, Y.~Zhang, Y.~Sun, W.~Wang, M.~Li, W.~Zhang, and X.~Lin, ``Approximate
  nearest neighbor search on high dimensional data—experiments, analyses, and
  improvement,'' \emph{TKDE}, vol.~32, no.~8, pp. 1475--1488, 2019.

\bibitem{Huang2015QueryAwareLH}
Q.~Huang, J.~Feng, Y.~Zhang, Q.~Fang, and W.~Ng, ``Query-aware
  locality-sensitive hashing for approximate nearest neighbor search,''
  \emph{PVLDB}, vol.~9, pp. 1--12, 2015.

\bibitem{Jiang2015ScalableGH}
Q.-Y. Jiang and W.-J. Li, ``Scalable graph hashing with feature
  transformation,'' in \emph{IJCAI}, 2015.

\bibitem{Muja2014ScalableNN}
M.~Muja and D.~G. Lowe, ``Scalable nearest neighbor algorithms for high
  dimensional data,'' \emph{TPAMI}, vol.~36, pp. 2227--2240, 2014.

\bibitem{malkov2020efficient}
Y.~A. Malkov and D.~Yashunin, ``Efficient and robust approximate nearest
  neighbor search using hierarchical navigable small world graphs,''
  \emph{TPAMI}, vol.~42, no.~04, pp. 824--836, 2020.

\bibitem{deng2022turl}
X.~Deng, H.~Sun, A.~Lees, Y.~Wu, and C.~Yu, ``Turl: Table understanding through
  representation learning,'' \emph{SIGMOD}, vol.~51, no.~1, pp. 33--40, 2022.

\bibitem{yin2020tabert}
P.~Yin, G.~Neubig, W.-t. Yih, and S.~Riedel, ``Tabert: Pretraining for joint
  understanding of textual and tabular data,'' in \emph{ACL}, 2020, pp.
  8413--8426.

\bibitem{yamada2020luke}
I.~Yamada, A.~Asai, H.~Shindo, H.~Takeda, and Y.~Matsumoto, ``Luke: Deep
  contextualized entity representations with entity-aware self-attention,'' in
  \emph{EMNLP}, 2020, pp. 6442--6454.

\bibitem{kim2022reweighting}
T.~Kim, C.~Park, J.~Hong, R.~Dua, E.~Choi, and J.~Choo, ``Reweighting strategy
  based on synthetic data identification for sentence similarity,'' in
  \emph{COLING}, 2022, pp. 4853--4863.

\bibitem{arabzadeh2021matches}
N.~Arabzadeh, A.~Bigdeli, S.~Seyedsalehi, M.~Zihayat, and E.~Bagheri, ``Matches
  made in heaven: Toolkit and large-scale datasets for supervised query
  reformulation,'' in \emph{CIKM}, 2021, pp. 4417--4425.

\bibitem{lim2022q2r}
S.~H. Lim and L.~Wynter, ``Q2r: A query-to-resolution system for
  natural-language queries,'' in \emph{NAACL}, 2022, pp. 353--361.

\bibitem{hall2003benchmarking}
M.~A. Hall and G.~Holmes, ``Benchmarking attribute selection techniques for
  discrete class data mining,'' \emph{TKDE}, vol.~15, no.~6, pp. 1437--1447,
  2003.

\bibitem{Paulsen2023SparklyAS}
D.~Paulsen, Y.~Govind, and A.~Doan, ``Sparkly: A simple yet surprisingly strong
  tf/idf blocker for entity matching,'' \emph{PVLDB}, vol.~16, pp. 1507--1519,
  2023.

\bibitem{Ester1996ADA}
M.~Ester, H.-P. Kriegel, J.~Sander, and X.~Xu, ``A density-based algorithm for
  discovering clusters in large spatial databases with noise,'' in
  \emph{SIGKDD}, 1996.

\bibitem{kriegel2011density}
H.-P. Kriegel, P.~Kr{\"o}ger, J.~Sander, and A.~Zimek, ``Density-based
  clustering,'' \emph{WIREs DMKD}, vol.~1, no.~3, pp. 231--240, 2011.

\bibitem{shopee-product-matching}
\BIBentryALTinterwordspacing
A.~Howard, C.~Liew, M.~Wong, and S.~Dane, ``Shopee - price match guarantee,''
  2021. [Online]. Available:
  \url{https://kaggle.com/competitions/shopee-product-matching}
\BIBentrySTDinterwordspacing

\bibitem{pennington2014glove}
J.~Pennington, R.~Socher, and C.~D. Manning, ``Glove: Global vectors for word
  representation,'' in \emph{EMNLP}, 2014, pp. 1532--1543.

\bibitem{min2020emogi}
S.~W. Min, V.~S. Mailthody, Z.~Qureshi, J.~Xiong, E.~Ebrahimi, and W.-m. Hwu,
  ``Emogi: efficient memory-access for out-of-memory graph-traversal in gpus,''
  \emph{PVLDB}, vol.~14, no.~2, pp. 114--127, 2020.

\bibitem{chen2020pangolin}
X.~Chen, R.~Dathathri, G.~Gill, and K.~Pingali, ``Pangolin: An efficient and
  flexible graph mining system on cpu and gpu,'' \emph{PVLDB}, vol.~13, no.~8,
  pp. 1190--1205, 2020.

\bibitem{elmagarmid2014nadeef}
A.~Elmagarmid, I.~F. Ilyas, M.~Ouzzani, J.-A. Quian{\'e}-Ruiz, N.~Tang, and
  S.~Yin, ``Nadeef/er: Generic and interactive entity resolution,'' in
  \emph{SIGMOD}, 2014, pp. 1071--1074.

\bibitem{singh2017synthesizing}
R.~Singh, V.~V. Meduri, A.~Elmagarmid, S.~Madden, P.~Papotti, J.-A.
  Quian{\'e}-Ruiz, A.~Solar-Lezama, and N.~Tang, ``Synthesizing entity matching
  rules by examples,'' \emph{PVLDB}, vol.~11, no.~2, pp. 189--202, 2017.

\bibitem{gokhale2014corleone}
C.~Gokhale, S.~Das, A.~Doan, J.~F. Naughton, N.~Rampalli, J.~Shavlik, and
  X.~Zhu, ``Corleone: Hands-off crowdsourcing for entity matching,'' in
  \emph{SIGMOD}, 2014, pp. 601--612.

\bibitem{wang2012crowder}
J.~Wang, T.~Kraska, M.~J. Franklin, and J.~Feng, ``Crowder: Crowdsourcing
  entity resolution,'' \emph{PVLDB}, vol.~5, no.~11, 2012.

\bibitem{konda2016magellan}
P.~Konda, S.~Das, A.~Doan, A.~Ardalan, J.~R. Ballard, H.~Li, F.~Panahi,
  H.~Zhang, J.~Naughton, S.~Prasad \emph{et~al.}, ``Magellan: toward building
  entity matching management systems over data science stacks,'' \emph{PVLDB},
  vol.~9, no.~13, pp. 1581--1584, 2016.

\bibitem{ebraheem2018distributed}
M.~Ebraheem, S.~Thirumuruganathan, S.~Joty, M.~Ouzzani, and N.~Tang,
  ``Distributed representations of tuples for entity resolution,''
  \emph{PVLDB}, vol.~11, no.~11, pp. 1454--1467, 2018.

\bibitem{zhao2019auto}
C.~Zhao and Y.~He, ``Auto-em: End-to-end fuzzy entity-matching using
  pre-trained deep models and transfer learning,'' in \emph{WWW}, 2019, pp.
  2413--2424.

\bibitem{peeters2021dual}
R.~Peeters and C.~Bizer, ``Dual-objective fine-tuning of bert for entity
  matching,'' \emph{PVLDB}, vol.~14, pp. 1913--1921, 2021.

\bibitem{wang2022sudowoodo}
R.~Wang, Y.~Li, and J.~Wang, ``Sudowoodo: Contrastive self-supervised learning
  for multi-purpose data integration and preparation,'' \emph{arXiv preprint
  arXiv:2207.04122}, 2022.

\bibitem{genossar2022flexer}
B.~Genossar, R.~Shraga, and A.~Gal, ``Flexer: Flexible entity resolution for
  multiple intents,'' \emph{arXiv preprint arXiv:2209.07569}, 2022.

\bibitem{miao2021rotom}
Z.~Miao, Y.~Li, and X.~Wang, ``Rotom: A meta-learned data augmentation
  framework for entity matching, data cleaning, text classification, and
  beyond,'' in \emph{SIGMOD}, 2021, pp. 1303--1316.

\bibitem{kasai2019low}
J.~Kasai, K.~Qian, S.~Gurajada, Y.~Li, and L.~Popa, ``Low-resource deep entity
  resolution with transfer and active learning,'' in \emph{ACL}, 2019, pp.
  5851--5861.

\bibitem{nafa2022active}
Y.~Nafa, Q.~Chen, Z.~Chen, X.~Lu, H.~He, T.~Duan, and Z.~Li, ``Active deep
  learning on entity resolution by risk sampling,'' \emph{Knowledge-Based
  Systems}, vol. 236, p. 107729, 2022.

\bibitem{jin2021deep}
D.~Jin, B.~Sisman, H.~Wei, X.~L. Dong, and D.~Koutra, ``Deep transfer learning
  for multi-source entity linkage via domain adaptation,'' \emph{PVLDB},
  vol.~15, no.~3, pp. 465--477, 2021.

\bibitem{kirielle2022transer}
N.~Kirielle, P.~Christen, and T.~Ranbaduge, ``Transer: Homogeneous transfer
  learning for entity resolution.'' in \emph{EDBT}, 2022, pp. 2--118.

\bibitem{gao2022clusterea}
Y.~Gao, X.~Liu, J.~Wu, T.~Li, P.~Wang, and L.~Chen, ``Clusterea: scalable
  entity alignment with stochastic training and normalized mini-batch
  similarities,'' in \emph{SIGKDD}, 2022, pp. 421--431.

\bibitem{liu2023unsupervised}
X.~Liu, J.~Wu, T.~Li, L.~Chen, and Y.~Gao, ``Unsupervised entity alignment for
  temporal knowledge graphs,'' in \emph{WWW}, 2023, pp. 2528--2538.

\bibitem{yao2022entity}
D.~Yao, Y.~Gu, G.~Cong, H.~Jin, and X.~Lv, ``Entity resolution with
  hierarchical graph attention networks,'' in \emph{Proceedings of the 2022
  International Conference on Management of Data}, 2022, pp. 429--442.

\bibitem{nentwig2016holistic}
M.~Nentwig, A.~Gro{\ss}, and E.~Rahm, ``Holistic entity clustering for linked
  data,'' in \emph{ICDM}.\hskip 1em plus 0.5em minus 0.4em\relax IEEE, 2016,
  pp. 194--201.

\end{thebibliography}

\end{document}